\RequirePackage{fix-cm}
\documentclass[smallextended]{svjour3}       
\smartqed  
\usepackage{graphicx}
\usepackage{setspace}

\usepackage{amssymb}
\usepackage{amsfonts}
\usepackage{amsmath}
\usepackage{enumerate}
\usepackage{multicol}

\newcommand\cF{{\cal F}}

\newcommand{\wh}{\widehat}

\def\bbr{{\mathbb R}}

\def\E{{\bf E}}

\def\P{{\bf P}}

\def\Chi{{\bf 1}}

\begin{document}

\title{Hedging problems for Asian options with  transactions costs 
}


\author{Serguei Pergamenchtchikov        \and
        Alena Shishkova 
}


\institute{Serguei Pergamenchtchikov  \at
              Laboratoire de Math\'ematiques Raphael Salem,
 UMR 6085 CNRS- Universit\'e de Rouen Normandie,  France
and
ILSSP\&QF, National Research Tomsk State University \\
              Tel.: 02 32 95 52 22\\
              Fax: 02 32 95 52 86\\
              \email{Serge.Pergamenchtchikov@univ-rouen.fr }           
           \and
           Alena Shishkova \at
             ILSSP\&QF, National Research Tomsk State University\\
              \email{alshishkovatomsk@gmail.com } 
}

\date{Received: date / Accepted: date}

\maketitle

\begin{abstract}
In this paper, we consider the problem of hedging Asian options in financial markets with transaction costs. For this we use the asymptotic hedging approach. The main task of asymptotic hedging in financial markets with transaction costs is to prove the probability convergence of the terminal value of the investment portfolio to the payment function when the number of portfolio revisions tends to be $n$ to infinity. In practice, this means that the investor, using such a strategy, is able to compensate payments for all financial transactions, even if their number increases unlimitedly.

\end{abstract}

\section{Introduction}
\label{intro}

For a trader or an investor the main task is not only the saving but also the multiplication of its capital. Many risks can be avoided with the help of one popular and very effective technique – hedging. The option is hedged to protect its value from the risk of price movement of the underlying asset in an unfavorable direction. To solve the hedging problem stochastic calculus methods are used which became a powerful tool used in practice in the financial world. Stochastic calculus is a well-developed branch of modern mathematics with a “correct” approach to analyzing complex phenomena occurring on world stock markets. 

In the modern theory and practice of options the paper written by Black and Scholes \cite{shishkova_Bl_Sch} has an important role. In this work the authors used economic knowledge in combination with PDE arguments which are similar to deriving the heat equation from the first physical principles.
	
In our paper we use a probabilistic approach and our main tool is representation theorem for the Wiener process. This theorem was formulated by J.M.C. Clark in \cite{shishkova_clark} also it can be obtained from the representation theorem stated in the paper \cite{shishkova_Ito} by K. Ito.

It should be noted that the task of options pricing and the construction of a hedging strategy is well studied for American and European options, for such derivatives there is a so-called delta strategy. But this technique is not enough developed for Asian options.  Exotic options became more in demand in the late 1980s and early 1990s and their trade became more active in the over-the-counter market. Soon in the commodity and currency markets, Asian options were becoming popular. 

Mathematically, the value of an Asian option is reduced to calculating the conditional mathematical expectation of a payment function. Many authors have studieded in their work the Asian pricing problem. H. Geman and M. Yor (1993) were among the first to consider derivatives are based on the average prices of underlying assets \cite{Geman_1993}. Using the Bessel processes authors found the value of the Asian option. Moreover, applying simple probabilistic methods they obtained the following results about these options: calculated moments of all orders of the arithmetic average of the geometric Brownian motion; obtained simple, closed form expression of the Asian option price when the option is “in the money”. The exact pricing of fixed-strike Asian options is a difficult task, since the distribution of the average arithmetic of asset prices is unknown when its prices are distributed lognormally. 

L. C. G. Rogers  and Z. Shi (1995) in their work \cite{Rog_Shi_1995} to compute the price of an Asian option used two different ways. Firstly, exploiting a scaling property, they reduced the problem to the problem of solving a parabolic PDE in two variables. Secondly, authors provided a lower bound which is so accurate that it is essentially the true price. J. Vecer (2001) observed that the Asian option is a special case of the option on a traded account and extended work \cite{Shr_Vec_2000} to the arithmetic average Asian option \cite{Vec_2001}. Using probabilistic techniques he established that the price of the Asian option is characterized by a simple one-dimensional partial differential equation which could be applied to both continuous and discreteaverage Asian option. J.Vecer and M.Xu (2004) studied pricing Asian options in a semimartingale model \cite{Vec_Xu_2004}. They showed that the inherently path dependent problem of pricing Asian options can be transformed into a problem without path dependency in the payoff function. Authors also showed that the price satisfies a simpler integro-differential equation in the case the stock price is driven by a process with independent increments, Levy process being a special case. Pricing Asian options under Levy processes also considered in \cite{Shir_2008,Fusai_2008,Albrec_2004}. 

When explicit valuation formulas are not available, sharp lower and upper bounds intervals for option
prices can be useful in improving the quality of the approximations adopted: some results in this direction are provided in the papers \cite{Simon_2000,Vanm_2006,Zeng_2016}. In \cite{Niel_2003} authors considered geometric average Asian option and showed that the lower and upper bounds can be expressed as a portfolio of delayed payment European call options. Pricing Asian options with stochastic volatility considered in \cite{Fouq_2003,Huba_2016,Shi_2014}.

A large number of works are connected with the numerical approach. Kemna and Vorst were among the first who solved the task \cite{Kemna_1990}. In their work the pricing strategy includes Monte Carlo simulation with elements of dispersion reduction and improves the pricing method. Furthermore, the authors showed that the price of an option with an average value will always be lower than of a standard European option. Carverhill and Clewlow \cite{Carv_1990} used a fast Fourier transform to calculate the density of the sum of random variables, as convolution of individual densities. Then the payoff function is numerically integrated against the density. In this direction other authors continued to work, applying to the calculations improved methods of numerical simulation \cite{Lap_Tem_2001,Fu_Mad_1999,Segh_Lid_2008}. Unfortunately, these methods do not provide information on the hedging portfolio.

In the articles above, the authors focus on calculating the value of the option, but do not consider in detail the hedging problem, use only general existence theorems. In works \cite{Albr_2003,Jac_1996}  authors consider the problem of hedging with the payoff
$$
f=\left(\frac{S_{t_1}+...+S_{t_n}}{n}-K\right)_+
$$
and use the moment recurrence technique, i.e. get recurrence equations. We can not use this technique, as we are considering the following payoff
$$
f_1=\left(\int\limits_0^1 S_udu-K\right)_+.
$$

General equations for the hedging strategy based on the martingal representation
$$
M_t=\E(f|\cF_t)=M_0+\int\limits_0^t\alpha_sdW_s,
$$
where
$$
\alpha_s=\langle M,W\rangle_s.
$$
This theory is well developed only for options whose payoffs depend only on the price at the last moment in time $f=f(S_T)$. And further, to compute a strategy, it is necessary to study only one random variable $S_T$, which is a geometric Brownian motion, i.e. density is known. In the case of Asian options, the payoff is a functional of the whole path and it is required to study the density of the integrals $\int\limits_0^t S_udu$ in order to calculate $\langle M,W\rangle_s$, that is, it is necessary to average over an infinite-dimensional distribution.

Yor and Dufrance \cite{Geman_1993,Dufr_2005} obtained the pricing in the explicit form, but they studied not the density, they considered the functional $f$, and for it they got a representation in the form of infinite series on a special orthogonal basis, which is not possible to study analiticaly the regularity properties. There is another method to study this properties of the density one can use the Brownian bridge, which proposed by Kabanov Yu.M. and Pergamenshchikov S.M. (2016) \cite{Kab_2016}. Using this method we construct the hedging strategy. This is made in \cite{shishkova_Shi_02_2018}. 

It is worth noting that the option pricing model in work \cite{shishkova_Bl_Sch} has an ideal character, i.e. it is assumed that it is friction-free market without costs. This theory is no longer true when we need to take into account transaction costs $\kappa_nJ_n=\kappa_0 n^{-\alpha}J_n$ because there is no  unimprovable hedge. Therefore  option pricing and replication with nonzero trading costs are different from that in the Black-Scholes setting.

Models with proportional transaction costs were considered as early as the 1970s. Magill and Constantinides \cite{shishkova_Mag_Const} suggested in 1976 the consumption-investment model which is generalization of the Merton model of 1973 \cite{shishkova_Merton}. However, the article written by H. Leland \cite{shishkova_Leland} in 1985 became more important for practical application. Leland's strategy provides an easy way to effectively eliminate the risks associated with transaction costs. This method is based on the idea that transaction costs can be offset by increasing the volatility parameter in the Black-Scholes strategy, that is the delta strategy obtained from a changed Black-Scholes equation with an appropriate modified volatility ensures an approximately complete replication as expected. The major goal in Leland's algorithm is to explore the asymptotic behavior of the hedging error (difference between the terminal value of portfolio and the payoff function) as the number of transaction goes to infinity. 

Leland suggested that if transaction costs are fixed, i.e. $\alpha=0$ then the value of the portfolio converges in probability to the payoff function as $n\rightarrow\infty$. He also suggested that this result will be true in the case of $\alpha=1/2$. Later this fact has proved by K. Lott  in his thesis \cite{shishkova_Lott}.   Later Yu. Kabanov and M. Safarian \cite{shishkova_Kab_Saf} extended Lott’s work to any $\alpha\in(0,1/2]$. Also they considered the case when $\alpha=0$, i.e. constant transaction cost. The authors proved that the hedging error admits a non-zero limit. The obtained result was used by H.Ahn end others \cite{shishkova_Ahn} for the hedging problem with transaction costs in general diffusion models. 

There are a lot of studies using Leland’s algorithm and extend it to various setting. For example, S. Pergamenshchikov in \cite{shishkova_perg} studied the convergence rate of approximation in the case of constant costs. He obtained technically difficult result since used nontrivial procedure. This result is important because it not only provides asymptotic information about the hedging error but also gives a reasonable way to solve the hedging problem, namely, the investor can get a portfolio whose final value exceeds the desired profit by choosing the appropriate value of the modified volatility.

The important result had been obtained by E. Lepinette \cite{shishkova_Lepin1} in the case of time-depending volatility models. He used a non-uniform interval splitting. Moreover, to obtain the asymptotically complet replication he modified the strategy, which is called Lepinette’s strategy, and proved that for $\alpha>0$ the portfolio value of this strategy converges in probability to the payoff and if $\alpha=0$ then the portfolio value of strategy converges in probability to the payoff plus two positive functions depending on payoff. To improve a rate of convergence E. Lepinette in \cite{shishkova_Lepin2} also used a non-uniform interval splitting and proved that for strategy suggested in \cite{shishkova_Lepin1} with $\alpha=0$ the approximation error multiplied by $n^{\beta}$ weakly converges to a centered mixed Gaussian variable as $n\rightarrow\infty$. 

Another way to enlarge application of Leland’s strategy is to consider the hedging problem with transaction costs in the models where the value of volatility depends on time and on the price of the stock, so-called the local volatility models. E.Lepinette and T.Tran \cite{shishkova_Lep_Tran} extended results obtained in \cite{shishkova_Lepin2} to this models. The proof of the result is really complicated, since the existence of a solution of a non-uniform parabolic Cauchy problem is nontrivial, if we adjust the volatility as well as in work \cite{shishkova_Lepin1}.

To extend the Leland’s approach many others authors considered different situations including more general contingent claims, more general price processes and etc. see \cite{shishkova_Ngu,shishkova_Den_Kab}. Thus Leland’s strategy has great importance in option pricing and the hedging problem due to it is easily implemented in practice.

Our goal is to extend this hedging methods for the hedging problem for the financial markets with transaction costs. To this end we use the approximative hedging approach proposed Leland, Kabanov, Safarian, Pergamenshchikov, Lepinette \cite{shishkova_Leland,shishkova_Kab_Saf,shishkova_perg,shishkova_Lepin1}. Note that is all this paper the hedging strategy is based on the delta-strategy. But for Asian option one need to change basic strategy, i.e. to pass frpm delta-strategy to Asian hedging strategy constructed in \cite{shishkova_Shi_02_2018}.

In this paper we study assymptotic property for the portfolio value with transaction cost in the Black-Scholes model with risky asset without drift and risk-free asset with interest rate $r=0$. We use the modification of Leland’s strategy. Main result of our study are obtined sufficient conditions, which provide assymptotic hedging.

\section{Market model}\label{sec:Mconds}

\subsection{Main condition}\label{subsec:Condit}
We consider the continuous time classical Black-Scholes model on financial market with risk-		free asset (bond) and risky asset (stock). For simplicity we suppose that the risk-free rate 		$r=0$, i.e. the bond price is constant over time $B_t=1$ throughout this article. Let $(\Omega,		\cF_1,(\cF_t)_{0\leq t\leq 1},\P)$ be the standard filtered probability space with $\cF_t=			\sigma(W_s,0\leq s\leq t)$ and $W$ is a Wiener process. The asset price process $S_t$ given by
		\begin{equation}\label{model}
			dS_t=\sigma S_tdW_t,\;\;\;0\leq t\leq 1
		\end{equation}
	and admit the following explicit form
		$$
			S_t=S_0 e^{\sigma W_t-\sigma^2 t/2}.
		$$
	Remark that $S_t$ is a martingale under measure $\P$. The model is considered on the interval 
	$[0,1]$ where 1 is a maturity of the Asian option with payoff function
		$$
			f_1=\left(\int\limits_0^1 S_udu-K\right)_+.
		$$

\subsection{Hedging problem}
\begin{definition}
The financial strategy $(\Pi_t)_{0 \leq t \leq 1}=(\beta_t,\gamma_t)_{0\leq t\leq 1}$ is called an admissible self-financing strategy if it is $\cF_t$-adapted, integrable with
		$$
			\int\limits_0^t(|\beta_t|+\gamma_t^2)dt<\infty
		$$	
	and the portfolio value is
		$$
			V_t=\beta_t+\gamma_t S_t=V_0+\int\limits_0^t\gamma_udS_u.
		$$
\end{definition}
	Here $V_0$ is an initial capital, $\beta_t$ and $\gamma_t$ are quantity of the risk-free asset 			and risk asset respectively.
	
	Suppose an investor operating on a (B,S)-- market solves the following "investment problem": using a self-financing portfolio at some predetermined point in time $1$, in the future bring its capital to $f_1$.  Obviously, the implementation of this goal depends on the initial capital $ x $ invested in the portfolio and on the investor strategy $ (\Pi_t)_{0 \leq t \leq 1} $ of portfolio reorganization used by the investor.
\begin{definition}
For a given $ x>0 $ and $ f_1$ a self-financing strategy is called a $ (x, f_1) $ - hedge if 
$$
\forall\omega\in\Omega,\; V_0^{\Pi}=x,\; V_1^{\Pi}\geq f_1 \;\;\text{a.s.}
$$
\end{definition}

\subsection{Hedging problem with the transaction costs. Leland strategy}
Let
$$
dS_t=\sigma S_tdW_t,\;\;\;0\leq t\leq 1
$$
and interest rate  is zero. Let us explain the key idea in the Leland's algorithm in the case European call option. We suppose that for each successful trade, traders are charged by a cost that is proportional to the trading volume with the cost coefficient $\kappa$. Here $\kappa$ is a positive constant defined by market moderators. We assume that the investor plans to revise his portfolio at dates $(t_i)=i/n$, where $n$ is the number of revisions.
 
 Under the presence of proportional transaction costs, it was proposed by \cite{shishkova_Leland} and then generalized by \cite{shishkova_Kab_Saf} that the volatility should be adjusted as 
\begin{equation}\label{sec:LelStr_eq1}
\wh{\sigma}^2=\sigma^2+\sigma\kappa n^{1/2-\alpha}\sqrt{8/\pi}
\end{equation}

in order to create an artificial increase in the option price $C(t,S_t)$ to compensate possible trading fees. This form is inspired from the observation that the trading cost $\kappa_n S_{t_i}|C_x(t_i,S_{t_i})-C_x(t_{i-1},S_{t_{i-1}})|$ in the interval of time $[t_{i-1},t_i]$ can be approximate by
\begin{equation}\label{sec:LelStr_eq2}
\kappa_n S_{t_{i-1}}C_{xx}(t_{i-1},S_{t_{i-1}})|\Delta S_{t_i}|\approx
\kappa_n \sigma S^2_{t_{i-1}}C_{xx}(t_{i-1},S_{t_{i-1}})\E|\Delta W_{t_i}|.
\end{equation}
For simplicity, we assume that the portfolio is revised at uniform grig $t_i=i/n,i=1,...,n$ of the option life interval $[0,1]$. Taking into account that $\E |\Delta W_{t_i}/(\Delta t_i)^{1/2}|=\sqrt{2/\pi}$ one approximates the last term in \eqref{sec:LelStr_eq2} by 
$$
\kappa_n \sigma \sqrt{2/\pi}(\Delta t_i)^{1/2} S^2_{t_{i-1}}C_{xx}(t_{i-1},S_{t_{i-1}}),
$$
which is the cost paid for portfolio readjustment in $[t_{i-1},t_i]$. Hence, by the standard argument of Black-Scholes (BS) theory, the option price inclusive of trading cost should satisfy
$$
C_{t}(t_{i-1},S_{t_{i-1}})\Delta t_i+\frac{1}{2} \sigma^2 S^2_{t_{i-1}} C_{xx}(t_{i-1},S_{t_{i-1}})\Delta t_i+\kappa_n \sigma \sqrt{2/\pi}(\Delta t_i)^{1/2} S^2_{t_{i-1}}C_{xx}(t_{i-1},S_{t_{i-1}})=0.
$$
Since $\Delta t_i=1/n$, one deduces that
$$
C_{t}(t_{i-1},S_{t_{i-1}})+\frac{1}{2} (\sigma^2+\kappa_n \sigma \sqrt{n8/\pi}) S^2_{t_{i-1}} C_{xx}(t_{i-1},S_{t_{i-1}})=0,
$$
which implies that the option price inclusive trading cost should be evaluated by the following modified-volatility version of the Black-Scholes PDE
$$
\wh{C}_{t}(t,x)+\frac{1}{2}\wh{\sigma}^2 x^2 \wh{C}_{xx}(t,x)=0,\;\;\;\wh{C}(1,x)=\max (x-K,0),
$$
where the adjusted volatility $\wh{\sigma}$ is defined by \eqref{sec:LelStr_eq1}.

To compensate transaction costs caused by hedging activities, the option seller is suggested to follow the Leland strategy defined by the piecewise process
$$
\gamma_t^n=\sum\limits_{i=1}^n \wh{C}_x(t_{i-1},S_{t_{i-1}})\Chi_{(t_{i-1},t_i]}(t).
$$
Then the portfolio value corresponding to this strategy at time $t$ defined by
$$
V_t^n=V_0+\int_0^t\gamma_u^n dS_u-\kappa_n \sum\limits_{i=1}^n S_{t_j}
|\gamma^n_{t_i}-\gamma^n_{t_{i-1}}|.
$$
\begin{definition}
Strategy $\gamma_t^n$ is called hedging if
$$
V_1^n\xrightarrow[n\rightarrow\infty]{\mathbf{P}}f_1.
$$
\end{definition}

\section{Definition of  strategies for the Asian options}\label{sec:Tr}

\subsection{Without transaction costs}\label{subsec:StrWithoutTrCst}
The hedging problem for the Asian call option with the terminal payoff $f_1$ is to choose the 			admissible self-financing strategy $(\beta_t,\gamma_t)$ such that 
		$$
			V_1=V_0+\int\limits_0^1\gamma_udS_u\geq f_1,\;\;\;\text{a.s.}
		$$
To construct a hedging strategy in the case of model (\ref{model}) apply the representation theorem for quadratic integrated martingale to the following martingale	
\begin{equation}\label{ch1s02_eq3}
M_t=\mathbf{E}(f_1|\mathcal{F}_t).
\end{equation}
We will find the square integrable process $(\alpha_t)_{0\leq t\leq 1}$ adapted w.r.t. 
$\mathcal{F}_t$  such that for all $t\in [0,1]$ 
\begin{equation}\label{ch1s02_eq3'}
M_t=M_0+\int_0^t\alpha_sdW_s.
\end{equation}
Clearly that
\begin{equation}\label{ch1s02_eq3''}
dM_t=\alpha_tdW_t.
\end{equation}
For coefficients $\alpha_t$ we use the following formula
$$
\langle M,W\rangle_t=\int_0^t \alpha_s ds,
$$
therefore
$$
\alpha_t=\frac{d}{dt}\langle M,W\rangle_t.
$$
Also the portfolio value satisfies the equality
\begin{equation}\label{ch1s02_eq3'''}
dV_t=\gamma_t dS_t=\gamma_t\sigma S_tdW_t.
\end{equation}
Equating (\ref{ch1s02_eq3''}) and (\ref{ch1s02_eq3'''}), we obtain the formulas for strategy
 $\Pi=(\beta_t,\gamma_t)_{0\leq t\leq 1}$ 
\begin{align}
\gamma_t&=\alpha_t/\sigma S_t,\label{ch1s02_eq4_2}
\\
\beta_t&=\mathbf{E}f_1+\int_0^t\alpha_sdW_s-\gamma_tS_t,\label{ch1s02_eq4_1}
\end{align} 

In our case the martingale has the following form
\begin{equation}\label{ch1s02_eq5}
M_{t}=\mathbf{E}(f_1)|\mathcal{F}_t^W)=\mathbf{E}\left(\left(\int_0^1
S_vdv-K\right)_{+}|\mathcal{F}_t^W\right),
\end{equation}
If $v\geq t$ then 
$$
S_v=S_{t}\exp\left\{\sigma (W_v-W_t)-\sigma^2(v-t)/2\right\}.
$$
It means that we can represent the integral in the equality \eqref{ch1s02_eq5} as
$$
\int_{0}^{1}S_{v}dv=\xi_{t}+S_{t}\eta_{t},
$$
where
$$
\xi_{t}=\int\limits_{0}^{t}S_{v}dv,\quad \eta_{t}=\int\limits_{t}^{1}\exp\left\{\sigma(W_v-W_t)-\sigma^2(v-t)/2\right\}dv,
$$
Note that $\xi_{t}$ is measurable w.r.t. $\mathcal{F}_{t},$ and $\eta_{t}$ is independent on
$\mathcal{F}_{t}$. Hence
\begin{equation}\label{ch1s02_eq6}
M_{t}=G(t,\xi_{t},S_{t}),
\end{equation}
here
$$
G(t,x,y)=\mathbf{E}\left(x+y\eta_{t}-K\right)_{+}.
$$

\begin{theorem}\label{ch1s02_th1}
The function $G(t,x,y)$ has the continuous derivatives
$$
\frac{\partial}{\partial t}G(t,x,y),\;\frac{\partial}{\partial x}G(t,x,y),\; \frac{\partial}{\partial y}G(t,x,y), \frac{\partial^2}{\partial y^2}G(t,x,y).
$$ 
\end{theorem}
The proof see in \cite{shishkova_Shi_02_2018}.

Since for any $t>0$  the process $(W_{t+u}-W_t)_{u\geq 0}$  is Wiener process then distribution of the random variable $\eta_t$ coincides with the distribution of the following random variable 
\begin{equation}\label{ch1s02_eq7}
\tilde{\eta}_v=\int_0^v\exp\{\sigma W_u-\sigma^2u/2\}du,
\end{equation}
Therefore
\begin{equation}\label{ch1s02_eq8}
G(t,x,y)=\mathbf{E}\left(x+y\tilde{\eta}_v-K\right)_{+}.
\end{equation}
Taking into account Theorem \ref{ch1s02_th1} and applying Ito's formula to the function $G(t,x,y)$ we obtain
\begin{equation}\label{ch1s02_eq9}
M_t=M_0+\int_0^t\left(G'_t(v,\xi_v,S_v)+G'_x(v,\xi_v,S_v)+\frac{\sigma^2S_v^2}{2}G^{''}_yy(v,\xi_v,S_v)\right)dv+ \tilde{M}_t,
\end{equation}
where 
$$
\tilde{M}_t=\sigma\int_0^tG'_y(v,\xi_v,S_v)S_vdW_v,
$$ 
$G'_t=\partial G/\partial t$ and other partial derivative similarly.
The quadratic characteristic is calculated by the formula
\begin{align*}
\langle M,W\rangle_t=\P-\lim_{n\rightarrow\infty}\sum_{j=1}^n
\mathbf{E}\left(\left(M_{t_j}-M_{t_{j-1}}\right)\left(W_{t_j}-W_{t_{j-1}}\right)|\mathcal{F}_{t_{j-1}}\right).
\end{align*}
We have that
$$
\int_0^t\left(G'_t(v,\xi_v,S_v)+G'_x(v,\xi_v,S_v)+\frac{\sigma^2S_v^2}{2}G^{''}_yy(v,\xi_v,S_v)\right)dv=0,
$$
since it is the continuous martingale. Then
\begin{align*}
\langle M,W\rangle_t=\sigma\langle \tilde{M},W\rangle_t=\sigma\int_0^tG'_y(v,\xi_v,S_v)S_vdv.
\end{align*}
Next, we find the formula for calculating martingale coefficients in (\ref{ch1s02_eq3'})
\begin{equation}\label{ch1s02_eq10}
\alpha_t=\sigma G'_y(t,\xi_t,S_t)S_t.
\end{equation}
Using(\ref{ch1s02_eq10}) in formulas (\ref{ch1s02_eq4_1}) and (\ref{ch1s02_eq4_2}), we obtain the hedging strategy
		$$
			\gamma_t=G'_y(t,\xi_t,S_t).
		$$	
For the obtained strategy $V_1=f_1$. Moreover $G(t,x,y)$ is the unique solution of the following 		equation
		\begin{equation}\label{shish_PDE1}
		\begin{cases}
		G'_t(t,x,y)+yG'_x(t,x,y)+\frac{\sigma^2}{2}y^2 G'_{yy}(t,x,y)=0
		\\
		G(1,x,y)=(x-K)_+.	
		\end{cases}
		\end{equation}

\subsection{With transaction costs}
Suppose that traders have to pay for a successful transaction some fee which is proportional to 		the trading volume. We assume that the cost proportion $\kappa_n=\kappa_0 n^{-\alpha}$. To 				compensate the transaction cost Leland \cite{shishkova_Leland} suggested to correct the 				volatility. The new parameter $\hat{\sigma}$ we have to put in the PDE (\ref{shish_PDE1}) and 			calculate the strategy again with a new volatility. Applying the Leland approach we modify the 			strategy as follows
		$$
		\gamma_t^n=\sum_{i=1}^n\hat{G}'_y(t_{j-1},\xi_{t_{j-1}},S_{t_{j-1}})\chi_{(t_{j-1},t_j]}(t),
		$$
	where $\hat{G}'_y(t,x,y)$ is the solution of the equation (\ref{shish_PDE1}) with parameter 
	$\hat{\sigma}$. Moreover $\hat{G}'_y(t,x,y)$ has the following form
		$$
			\hat{G}'_y(t,x,y)=\int_b^{\infty}z\hat{q}(v,z)dz,\;\;\;b=\frac{(K-x)_+}{y}
		$$
	here $\hat{q}(v,z)$ is a density of random variable $\tilde{\eta}_v$ with new parameter 
	$\hat{\sigma}$ and given by
		$$
			\hat{q}(v,z)=\E\left(\frac{\varphi_{0,1}(\hat{a}(t,z))}{\hat{K}(v,\hat{a}(t,z))}\right),
		$$
		$$
			\hat{K}(v,\hat{a}(t,z))=\hat{\sigma}\int_0^v u
			\exp\{\hat{\sigma}\tilde{W}_u-\frac{\hat{\sigma}^2 u}{2}+\hat{\sigma} u \hat{a}(t,z)\},
			\;\;\;\;\tilde{W}_u=W_u-uW_1.
		$$
	This form of density has been received in the article \cite{shishkova_Shi_02_2018}. The 
	portfolio value at $t$ with the initial capital 
	$V_0=\hat{G}(0,\xi_0,S_0)$ has the form
		\begin{equation}\label{shish_eq_capital1}
			V_t^n=\hat{G}(0,\xi_0,S_0)+\int_0^t\gamma_u^ndS_u-\kappa_nJ_n,
		\end{equation}
	where the total trading volume is given by
		$$
			J_n=\sum\limits_{j=1}^nS_{t_j}|\gamma^n_{t_j}-\gamma^n_{t_{j-1}}|.
		$$
	
	In order to keep the hedging strategy it is necessary to satisfy the following condition
		$$
			V_1^n\xrightarrow[n\rightarrow\infty]{\mathbf{P}}f_1.
		$$
    For this we need to consider a hedging error $V_1^n-f_1$. By Ito formula we have
		\begin{align*}
			G(t,\xi_t,S_t)=G(0,\xi_0,S_0)&+\int_0^t\left(G'_t(u,\xi_u,S_u)+G'_x(u,\xi_u,S_u)S_u+					\frac{\sigma^2S_u^2}{2}G''_{yy}(u,\xi_u,S_u)\right)du\\
			&+\int_0^tG'_y(u,\xi_u,S_u)\sigma S_udW_u,
		\end{align*}
	since $G'_t(t,\xi_t,S_t)+G'_x(t,\xi_t,S_t)S_t+\frac{\sigma^2S_t^2}{2}G''_{yy}(t,\xi_t,S_t)=0$  then
		$$
			G(t,\xi_t,S_t)=G(0,\xi_0,S_0)+\int_0^tG'_y(u,\xi_u,S_u)\sigma S_udW_u=G(0,\xi_0,S_0)+					\int_0^tG'_y(u,\xi_u,S_u)dS_u.
		$$
	Condition of replication is
		$$
			V_1=f_1=\left(\int_0^1S_udu-K\right)_+.
		$$
	Since by the construction of the strategy  $V_1=G(1,\xi_1,S_1)$ then
		$$
			f_1=G(0,\xi_0,S_0)+\int_0^t\gamma_udS_u,
		$$
	where $\gamma_t=G'_y(t,\xi_t,S_t)$. Thus, taking into account (\ref{shish_eq_capital1}) we have 
		\begin{align*}
			V_1^n-f_1=\hat{G}(0,\xi_0,S_0)-G(0,\xi_0,S_0)+\int_0^t(\gamma_u^n-\gamma_u)dS_u-						\kappa_nJ_n.
		\end{align*}
	
	Since $G(1,\xi_1,S_1)=\hat{G}(1,\xi_1,S_1)=(x-K)_+$, i.e. the same boundary condition we can 			write the following equality
		\begin{align*}
			\hat{G}(0,\xi_0,S_0)-G(0,\xi_0,S_0)&=(G(1,\xi_1,S_1)-G(0,\xi_0,S_0))-(\hat{G}(1,						\xi_1,S_1)-\hat{G}(0,\xi_0,S_0))\\[3mm]
			&=\int_0^1G'_y(u,\xi_u,S_u)dS_u+\frac{\hat{\sigma}^2-\sigma^2}{2}\int_0^1\hat{G}''_{yy}					(u,\xi_u,S_u)S_u^2du
			\\[3mm]
			&-\int_0^1\hat{G}'_y(u,\xi_u,S_u)dS_u,
		\end{align*}
	By Ito's formula we have
		$$
			\hat{G}(1,\xi_1,S_1)=\hat{G}(0,\xi_0,S_0)+\frac{\sigma^2-\hat{\sigma}^2}{2}								\int_0^1\hat{G}''_{yy}(u,\xi_u,S_u)S_u^2du+\int_0^1\hat{G}'_y(u,\xi_u,S_u)dS_u
		$$
	and
		$$
			G(1,\xi_1,S_1)=G(0,\xi_0,S_0)+\int_0^1G'_y(u,\xi_u,S_u)dS_u.
		$$
	Then
		$$
			V_1^n-f_1=\int_0^1(G'_y(t,\xi_t,S_t)-\hat{G}'_y(t,\xi_t,S_t))dS_t+\frac{\hat{\sigma}^2-					\sigma^2}{2}\int_0^1\hat{G}''_{yy}(t,\xi_t,S_t)S_t^2dt+\int_0^1(\gamma_t^n-\gamma_t)dS_t-				\kappa_nJ_n
		$$
	Finally we obtain
		$$
			V_1^n-f_1=\int_0^1(\gamma_t^n-\hat{\gamma}_t)dS_t+\frac{\hat{\sigma}^2-\sigma^2}{2}						\int_0^1\hat{G}''_{yy}(t,\xi_t,S_t)S_t^2dt-\kappa_nJ_n,
		$$
	because $\int_0^1G'_y(t,\xi_t,S_t)dS_t=\int_0^1\gamma_tdS_t$ and $\hat{\gamma}_t=\hat{G}'_y(t,			\xi_t,S_t)$. 

\section{Option value analysis}
The option cost is defined as
$$
C_0=G(0,0,S_0)=\int\limits_{b_0}^{+\infty}(zS_0-K)_+q(1,z)dz,\;\;\;\;b_0=K/S_0.
$$
Recall that
$$
G(t,x,y)=\E(x+y\eta_v-K)_+=\int\limits_{b}^{+\infty}(x+yz-K)_+ q(v,z)dz,
$$
where $b=(K-x)_+/y,\;\;\;v=1-t$ and $q(v,z)$ is the density of the random variable 
$$
\eta_v=\int\limits_0^v\exp\left\lbrace \sigma W_u-\frac{\sigma^2}{2}u \right\rbrace du
$$ 
and given by
$$
q(v,z)=\E\frac{\varphi_{0,1}(a(v,z))}{K(v,a)}.
$$
Here $\varphi_{0,1}(a)$ is the Gaussian density and $a(v,z)$ has an implisit form
$$
z=\int\limits_0^1\exp\left\lbrace \sigma W_u-\sigma uW_1-\frac{\sigma^2}{2}u+\sigma ua(v,z) \right\rbrace du.
$$
After we have introduced the transaction costs and changed the volatility as
$$
\hat{\sigma}^2=\sigma^2+\sigma\sqrt{\frac{8}{\pi}}\kappa_n\sqrt{n}
$$ 
we obtain that the cost of option is equal
$$
\wh{C}_0=\int\limits_{b_0}^{+\infty}(zS_0-K)_+\wh{q}(1,z)dz.
$$
There are three variants of changes in value of option. 

1) Case $\wh{\sigma}\rightarrow\sigma$ 

if $\kappa_n=\kappa_0 n^{-1/2}$ and $\kappa_0\rightarrow 0$ then  it is obvious that
$$
\wh{C}_0\rightarrow C_0.
$$

2) Case $\wh{\sigma}^2=\sigma^2+\sigma\sqrt{\frac{8}{\pi}}\kappa_n\sqrt{n}$ with $\kappa_n=1/\kappa_0$. In this case, the hedging will be, but the value of the option will increase by a constant $\sigma\sqrt{\frac{8}{\pi}}$.

3) Case $\wh{\sigma}\rightarrow +\infty$ 

if $\kappa_n=o(n^{-\alpha})$ and $\alpha\geq 2/5$. Then we obtain the strategy "buy and hold"
$$
\wh{C}_0\rightarrow S_0.
$$
It is proved in Proposition \ref{Propo_Cost_opt}.
\begin{proposition}\label{Propo_Cost_opt}
Let $\rho(u)=\exp\left\lbrace \wh{\sigma} W_u-\frac{\wh{\sigma}^2}{2}u \right\rbrace $ and 
$\rho(u)\rightarrow 0$ as $\wh{\sigma}\to\infty$. Then 
$$
\wh{\eta}_1=\int\limits_0^1\rho(u) du\xrightarrow[\wh{\sigma}\rightarrow\infty]{\mathbf{P}}0
$$
and
$$
\wh{C}_0\xrightarrow[\wh{\sigma}\to\infty]{}S_0.
$$
\end{proposition}

\begin{proof}
First of all we will prove that $\wh{\eta}_1\xrightarrow[\wh{\sigma}\rightarrow\infty]{\mathbf{P}}0
$. Represent it like
$$
\wh{\eta}_1=\wh{\eta}_{1,1}+\wh{\eta}_{1,2},
$$
where
$$
\wh{\eta}_{1,1}=\int\limits_0^{\delta}\rho(u) du,\;\;\;\wh{\eta}_{1,2}=\int\limits_{\delta}^1\rho(u) du.
$$
We choose $\delta$ so that it tends to zero not very quickly, for example 
$\delta=1/\sqrt{\wh{\sigma}}$. Then
$$
\E\wh{\eta}_{1,1}=\int\limits_0^{\delta}\E\rho(u) du=\delta\xrightarrow[\wh{\sigma}\to\infty]{}0.
$$
For the second termvwe can use estimate
\begin{align*}
\max\limits_{\delta\leq u\leq 1}\rho(u)
&\leq \exp\left\lbrace \wh{\sigma}\max\limits_{0\leq u\leq 1} W_u-\frac{\wh{\sigma}^2}{2}\delta \right\rbrace
\\[3mm]
&\leq \exp\left\lbrace \wh{\sigma}\max\limits_{0\leq u\leq 1} W_u-\frac{\wh{\sigma}^{3/2}}{2}\delta \right\rbrace\xrightarrow[\wh{\sigma}\to\infty]{}0 \;\;\;\text{a.s.}
\end{align*}
Thus $\E\wh{\eta}_{1,2}\xrightarrow[\wh{\sigma}\to\infty]{}0$. We have
\begin{align*}
\wh{C}_0&=\E(S_0\wh{\eta}_1-K)_+
\\[3mm]
&=\E(S_0\wh{\eta}_1-K)_+\Chi_{\{\wh{\eta}_1>K/S_0\}}+
\E(S_0\wh{\eta}_1-K)_+\Chi_{\{\wh{\eta}_1\leq K/S_0\}}
\\[3mm]
&=\E(S_0\wh{\eta}_1-K)\Chi_{\{\wh{\eta}_1>K/S_0\}}
=S_0\E\wh{\eta}_1\Chi_{\{\wh{\eta}_1>K/S_0\}}-K\P(\wh{\eta}_1>K/S_0)
\end{align*}
The last probabilities tends to zero therefore
\begin{align*}
\wh{C}_0=S_0\E\wh{\eta}_1\Chi_{\{\wh{\eta}_1>K/S_0\}}.
\end{align*}
If represent the indicator as $\Chi_{\{\wh{\eta}_1>K/S_0\}}=1-\Chi_{\{\wh{\eta}_1\leq K/S_0\}}$ then
$$
\E\wh{\eta}_1\Chi_{\{\wh{\eta}_1>K/S_0\}}=\E\wh{\eta}_1-\E\wh{\eta}_1\Chi_{\{\wh{\eta}_1\leq K/S_0\}}=1-\E\wh{\eta}_1\Chi_{\{\wh{\eta}_1\leq K/S_0\}}
$$
Since $\wh{\eta}_1$ is bounded and goes to zero then by Lebesgue's theorem on majorized convergence 
$\E\wh{\eta}_1\Chi_{\{\wh{\eta}_1\leq K/S_0\}}\rightarrow 0$. 
\end{proof}

\section{Properties of the density $q(v,z)$}\label{sec:DN}
To exlore the distribution of the random variable $\tilde{\eta}_v$ we introduce the notation of  Brownian bridge.
\begin{definition}
Coming from zero and coming to $a\in\mathbb{R}$ the Brownian bridge $(B_t^a)_{0\leq t\leq T}$  is the Gaussian process such that 
$$
B_t^a=W_t-t W_1+ t a,
$$
where $a$ -- some constant.
\end{definition}
Conditional distributions are calculated for a fixed finite value of the Wiener process using this process, i.e. for any function $L:C[0,1]\rightarrow\mathbb{R}$ and for any Borel set $\Gamma$
$$
\mathbf{P}(L(W_t)_{0\leq t\leq 1}\in\Gamma|W_t=a)=\mathbf{P}(L(B_t^a)_{0\leq t\leq 1}\in\Gamma).
$$
\begin{proposition}\label{ch1s03_prop1}
For any $0\leq t\leq 1$  the random variable $\tilde{\eta}_v$ has a distribution density.
\end{proposition}

\begin{proof}
Let $Q$  -– some bounded function $\mathbb{R}\rightarrow\mathbb{R}$. In our case 
$$
\mathbf{E}Q(\tilde{\eta}_v)=\mathbf{E}(\mathbf{E}(Q(\tilde{\eta}_v)|W_1))=\int_{\mathbb{R}}Q(F(v,a))\varphi(a)da,
$$
where
$$
F(v,a)=\int_0^v\exp\{\sigma W_u-\sigma u W_1-\sigma^2 u/2+\sigma ua\}du,
$$
Next we make the change of variable $z=F(v,a)$, i.e. we introduce the function $a=a(v,z)$  as
\begin{equation}\label{ch1s03_eq11}
z=F(v,a(v,z)).
\end{equation}
It means that
$$
a'_z(v,z)=\frac{1}{K(v,a(v,z))},
$$
where
\begin{equation}\label{ch1s03_eq12}
K(v,a)=F'_a(v,a)=\sigma\int_0^v u\exp\{\sigma W_u-\sigma u W_1-\sigma^2 u/2+\sigma ua\}du.
\end{equation}
Then
$$
\mathbf{E}Q(\tilde{\eta}_v)=\int_0^{\infty}Q(z)q(v,z)dz,
$$
here
$$
q(v,z)=\mathbf{E}\frac{\varphi(a)}{K(v,a)},\;\;\;\;\;\varphi(\cdot)\sim\mathcal{N}(0,T).
$$
Thus the density of the random variable $\tilde{\eta}_v$ has the form
\begin{equation}\label{ch1s03_eq13}
q(v,z)=\frac{\varphi(a)}{K(v,a)}.
\end{equation}
\end{proof}

Next we will use the following propositions.  
	\begin{proposition}\label{shish_density_estim_part1}		
		For $v_*=\min (\sigma^2 v,1)$ and some constants $\tilde{c}>0$ and $\kappa>0$
			\begin{equation*}\label{eq_estim_q}
				q(v,z)\leq\frac{\tilde{c}\sigma^3}{v_*^2}\left(
				\exp\left\lbrace-\frac{\kappa}{ \sigma^2 v}(\ln(z/v))^2\right\rbrace 
				\Chi_{\{z>v\}}+\Chi_{\{z\leq v\}}\right),
			\end{equation*}
			
			\begin{equation*}\label{eq_estim_q_z}
				|q_z(v,z)|\leq\frac{\tilde{c}\sigma^7}{v_*^4}\left(\exp\left\lbrace-\frac{\kappa}{ \sigma^2 v}
				(\ln(z/v))^2\right\rbrace \Chi_{\{z>v\}}+\Chi_{\{z\leq v\}}\right),
			\end{equation*}
	\end{proposition}

	\begin{proposition}\label{shish_density_estim_part2}
	For $v_*=\min (\sigma^2 v,1)$ and some constants $\tilde{c}>0$ and $\kappa>0$
			\begin{equation*}\label{eq_estim_q_z}
				|q_v(v,z)|\leq \frac{\tilde{c}\sigma^7}{v_*^4}\left(\Chi_{\{z\leq v\}}+\exp\left\lbrace-\frac{\kappa}{ \sigma^2 v}
				(\ln(z/v))^2\right\rbrace \Chi_{\{z>v\}}\right).
			\end{equation*}
			
	\end{proposition}
See proofs in Appendix.	

\section{Asymptotic hedging}\label{sec:MR1}
Recall that an option seller should increase volatility in order to compensate for transaction costs. Choose a new volatility parameter
	\begin{equation}\label{shish_sigma_hat}
		\hat{\sigma}^2=\sigma^2+\sigma\sqrt{n}\kappa_n\sqrt{\frac{8}{\pi}},\;\;\;\kappa_n=\kappa_0 n^{-\alpha}.
	\end{equation}
Then the following theorem holds.
	\begin{theorem}
		For $\alpha=1/2$ in the equation (\ref{shish_sigma_hat}) the portfolio value $V_1^n$ converges in probability to the payout function $f_1$ as $n\rightarrow\infty$.
	\end{theorem}
	
	\begin{proof}
		We have an expression for a hedging error
			$$
			V_1^n-f_1=\int_0^1(\gamma_t^n-\hat{\gamma}_t)dS_t+\frac{\hat{\sigma}^2-\sigma^2}{2}						\int_0^1\hat{G}''_{yy}(t,\xi_t,S_t)S_t^2dt-\kappa_nJ_n.
			$$
		Since
		$$
		\gamma_t^n=\sum_{i=1}^n\hat{G}'_y(t_{j-1},\xi_{t_{j-1}},S_{t_{j-1}})								\chi_{(t_{j-1},t_j]}(t)
		$$ 
		and $\hat{\gamma}_t=\hat{G}'_y(t,\xi_{t},S_{t})$ uniformly continuous on the segment $[0,1]$, it is obvious that the first term tends to zero as $n\rightarrow\infty$ and it remains only to verify that
		$$
			\frac{\hat{\sigma}^2-\sigma^2}{2}\int_0^1\hat{G}''_{yy}(t,\xi_t,S_t)S_t^2dt-\kappa_nJ_n					\xrightarrow[n\rightarrow\infty]{\mathbf{P}}0.
		$$
		First, we’ll evaluate $\kappa_nJ_n$. We introduce the notation
		$$
			H(t_{j},\xi_{t_{j}},S_{t_j})=\hat{G}_y(t_{j},\xi_{t_{j}},S_{t_j}).
		$$
		Then
		$$
			\kappa_nJ_n=\kappa_n\sum_{j=1}^n S_{t_j}|H(t_{j},\xi_{t_{j}},S_{t_j})-					H(t_{j-1},\xi_{t_{j-1}},S_{t_{j-1}})|.
		$$
		Add and subtract the term $|H(t_{j-1},\xi_{t_{j-1}},S_{t_j})-H(t_{j-1},								\xi_{t_{j-1}},S_{t_{j-1}})|$ and represent $\kappa_nJ_n$ as
		$$
			\kappa_nJ_n=A_n^{(1)}+A_n^{(2)},
		$$
		where
		$$
			A_n^{(1)}=\kappa_n\sum_{j=1}^n S_{t_j}|H(t_{j-1},\xi_{t_{j-1}},S_{t_{j}})-
			H(t_{j-1},\xi_{t_{j-1}},S_{t_{j-1}})|
		$$
		and
		$$
			A_n^{(2)}=\kappa_n\sum_{j=1}^n S_{t_j}\left(|H(t_{j},									
			\xi_{t_{j}},S_{t_j})-H(t_{j-1},\xi_{t_{j-1}},S_{t_{j-1}})|-|H(t_{j-1},									\xi_{t_{j-1}},S_{t_{j}})-H(t_{j-1},\xi_{t_{j-1}},S_{t_{j-1}})|\right).
		$$
		Using the fact $||x|-|y||\leq |x-y|$, we obtain
		$$
			|A_n^{(2)}|\leq \kappa_n\sum_{j=1}^n S_{t_j}
			|H(t_{j},\xi_{t_{j}},S_{t_j})-H(t_{j-1},\xi_{t_{j-1}},S_{t_{j}})|=B_n.
		$$
		In the section \ref{sec_dop_lem_small_volatil} we proved that  $\P-\lim\limits_{n\rightarrow\infty}B_n=0$, see Lemma  \ref{shish_lemma1_B_n_small_volatil}. Thus, further we need to consider only $A_n^{(1)}$. Recall that according to the Taylor formula we can write
		$$
			H(t_{j-1},\xi_{t_{j-1}},S_{t_{j}})=H(t_{j-1},\xi_{t_{j-1}},S_{t_{j-1}})+H_y(t_{j-1},					\xi_{t_{j-1}},S_{t_{j-1}})(S_{t_{j}}-S_{t_{j-1}})+o(n^{-2})
		$$
		and represent $A_n^{(1)}$ as
		\begin{align*}
			A_n^{(1)}&=\kappa_n\sum_{j=1}^n S_{t_j}\left(|H(t_{j-1},								
			\xi_{t_{j-1}},S_{t_{j}})-H(t_{j-1},\xi_{t_{j-1}},S_{t_{j-1}})|-|H_y(t_{j-1},							\xi_{t_{j-1}},S_{t_{j-1}})||S_{t_{j}}-S_{t_{j-1}}|\right)\\
			&+\kappa_n\sum_{j=1}^n S_{t_j}|H_y(t_{j-1},\xi_{t_{j-1}},S_{t_{j-1}})|
			|S_{t_{j}}-S_{t_{j-1}}|.
		\end{align*}
		Next we denote
		$$
			D_n^{(2)}=\kappa_n\sum_{j=1}^n S_{t_j}\left(|H(t_{j-1},								
			\xi_{t_{j-1}},S_{t_{j}})-H(t_{j-1},\xi_{t_{j-1}},S_{t_{j-1}})|-|H_y(t_{j-1},							\xi_{t_{j-1}},S_{t_{j-1}})||S_{t_{j}}-S_{t_{j-1}}|\right)
		$$
		and in Lemma \ref{shish_lem2_D_n_small_volatil} we will prove $\P-\lim\limits_{n\rightarrow\infty}D_n^{(2)}=0$. Thus, we have
		$$
			\kappa_nJ_n\approx\kappa_n\sum_{j=1}^n S_{t_j}
			|H_y(t_{j-1},\xi_{t_{j-1}},S_{t_{j-1}})||S_{t_{j}}-S_{t_{j-1}}|
		$$
		and by Lemma \ref{shish_lem3_converg_small_volatil} we obtain that
		$$
		\kappa_nJ_n\xrightarrow[n\rightarrow\infty]{\mathbf{P}}\kappa_n\sqrt{n}\sigma
		\sqrt{\frac{2}{\pi}}\int_0^1\hat{G}''_{yy}(t,\xi_t,S_t)S_t^2dt.
		$$
		Then
		$$
		\frac{\hat{\sigma}^2-\sigma^2}{2}\int_0^1\hat{G}''_{yy}(t,\xi_t,S_t)S_t^2dt-\kappa_nJ_n 				\xrightarrow[n\rightarrow\infty]{\mathbf{P}}0
		$$
	\end{proof}

\section{Simulations}

To build a hedging strategy, we need to calculate the coefficients $(\alpha_t)_{0\leq t\leq 1}$. First we compute the function $G(t,x,y)$ , for this we simulate L random variables $\eta_t^j$. We take the time step which is equal
$$
dt=1/N,
$$
where $N$ -- number of partitions. The mathematical expectation is calculated by the Monte Carlo method. We get the computational formula
\begin{equation}\label{ch1s05_eq28}
\eta_t=\frac{1-t}{N}\sum\limits_{k=1}^N\exp\left\lbrace\sigma W_k(1-t)-\sigma^2 \frac{(1-t)k}{2N}\right\rbrace.
\end{equation}
Then for the function $G(t,x,y)$ we obtain the expression
\begin{equation}\label{ch1s05_eq29}
G(t,x,y)\approx\frac{1}{L}\sum\limits_{j=1}^L(x+y\eta_t^j-K)_+.
\end{equation}
To calculate the partial derivative $G'_y(t,x,y)$ we use the following formula
\begin{equation}\label{ch1s05_eq30}
\frac{\partial}{\partial y}G(t,x,y)=\frac{G(t,x,y+\delta)-G(t,x,y)}{\delta},\;\;\;\delta=0,0001.
\end{equation}
Before proceeding to the calculation of the coefficients $(\alpha_t)_{0\leq t\leq 1}$ we write the calculation formulas for $(\xi_t)_{0\leq t\leq 1}$ and $(S_t)_{0\leq t\leq 1}$.
$$
\xi_t=\frac{S_0 t}{N}\sum\limits_{k=1}^N\exp\left\lbrace\sigma W_k(t)-\sigma^2 \frac{tk}{2N}\right\rbrace
$$
and
$$
S_t=S_0\exp\left\lbrace\sigma W_t-\sigma^2\frac{t}{2}\right\rbrace.
$$

Next, we find the coefficients $(\alpha_t)_{0\leq t\leq 1}$ and build a strategy $\Pi=(\beta_t,\gamma_t)_{0\leq t\leq 1}$.

Consider the implementation of the asset process and hedging strategies for  $\sigma=0.05$ 
 with $S_0=100$ and $N=100$, we obtain the following results.
 
 \begin{figure}[h!]
\center{\includegraphics[scale=0.5]{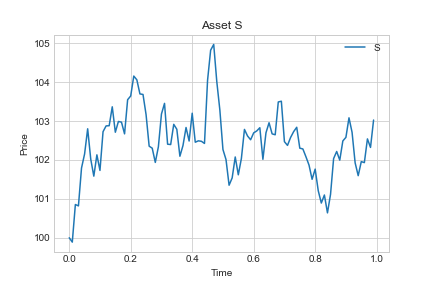}}
\caption{Asset price in a market with volatility $\sigma=0.05$.}
\label{fig1}
\end{figure}
\vspace{5mm}

\begin{figure}[h!]
\center{\includegraphics[scale=0.5]{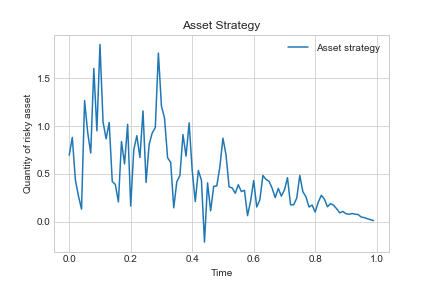}}
\caption{The quantity $(\gamma_t)_{0\leq t\leq 1}$ of risky asset $(S_t)_{0\leq t\leq 1}$ in a hedging strategy for an Asian market option with volatility $\sigma=0.05$.}
\label{fig2}
\end{figure}
\vspace{5mm}

\begin{figure}[h!]
\center{\includegraphics[scale=0.5]{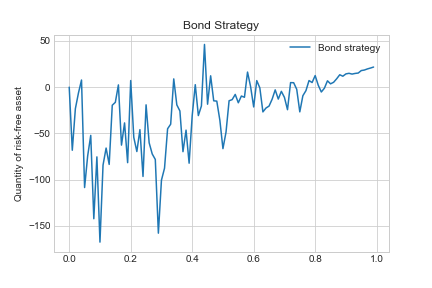}}
\caption{The quantity $(\beta_t)_{0\leq t\leq 1}$ of riskless asset  $(B_t)_{0\leq t\leq 1}$ in a hedging strategy for an Asian market option with volatility $\sigma=0.05$.}
\label{fig3}
\end{figure}
\vspace{5mm}
 
 Let us compare the value of the terminal portfolio and the payoff function for a different number of partitions $N$ with parameters $\sigma=0.1$, $S_0=100$, $K=50$.
\begin{table}[h!]
\caption{The terminal portfolio $X_1$ and the pauoff $f_1$}
\label{tabular:optpr2}
\begin{center}
\begin{tabular}{|c||*{6}{c|}}
\hline
$N$ & 20 & 50 & 100 & 200 & 500 & 1000 \\
\hline
$X_1$ & 74.1408 & 37.7123 & 22.76484379 & 184.6879 & 50.0242 &  10.220968  \\
 \hline
$f_1$ & 47.8504 & 49.5922  & 22.4617  & 61.5188 & 52.9817  &  6.37918036  \\
 \hline
\end{tabular}
\end{center}
\end{table}
\vspace{5mm}

The value of the option is calculated by the formula
$$
C_0\approx\frac{1}{L}\sum\limits_{j=1}^L\left(S_0\sum\limits_{k=1}^N dt* \exp\{\sigma W_k(1)-\sigma^2 t_k/2\}\right)_+.
$$

Consider the simulation results for $S_0=100$, $t\in[0,1]$, \linebreak ${L=500 000}$ and $n=1000$.

\begin{figure}[h!]
\center{\includegraphics[scale=0.7]{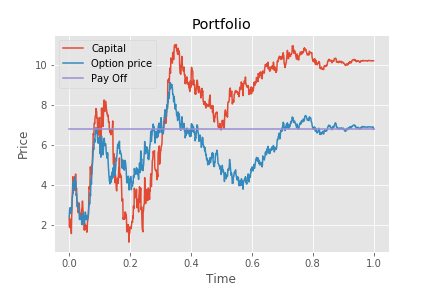}}
\caption{Graphs of option value, investor's capital and Asian option payoff function.}
\label{fig:portfolio}
\end{figure}
\vspace{5mm}

In Fig. \ref{fig:portfolio}time changes on the abscissa axis, and corresponding values on the ordinate axis. We see that at every moment in time, the trajectory of the option value almost repeats the trajectory of investor capital, which is natural for the hedging task. The size of the terminal portfolio exceeds the payoff function, which confirms that the strategy is hedging. 

Investigating the behavior of the option value depending on the initial stock price $S_0$, strike price $K$ and volatility $\sigma$, we obtain the following results.

\begin{table}[h!]
\caption{The dependence of the value of an Asian option on the volatility parameter with $K=S_0$}
\label{tabular:optpr2}
\begin{center}
\begin{tabular}{|c||*{7}{c|}}
\hline
$\sigma$ & 0.01 & 0.05 & 0.1 & 0.5 & 1 & 1.5 & 2\\
\hline
$C_0$ & 0.229 & 1.371 & 2.303 & 11.346 & 22.473 & 32.941 & 42.466 \\
 \hline
\end{tabular}
\end{center}
\end{table}

\begin{table}[h!]
\caption{The dependence of the value of an Asian option on the volatility parameter with $K=S_0/2$}
\label{tabular:optpr3}
\begin{center}
\begin{tabular}{|c||*{7}{c|}}
\hline
$\sigma$ & 0.01 & 0.05 & 0.1 & 0.5 & 1 & 1.5 & 2\\
\hline
$C_0$ & 50.115 & 50.201 & 50.107 & 50.055 & 51.669 & 55.832 & 59.443 \\
 \hline
\end{tabular}
\end{center}
\end{table}

Calculating the value of the option without costs and wiyh costs when parameters $\sigma = 0.05 $ and $S_0=100, K=70$ for a different number of portfolio revisions, we obtain the following result.
\begin{figure}[h!]
\center{\includegraphics[scale=0.6]{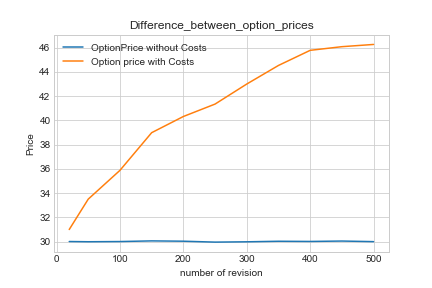}}
\caption{The value of the option in the market with transaction costs and without costs.}
\label{fig:Difference_option_prices}
\end{figure}
\vspace{5mm}

We have investigated the behavior of the hedging error $V_1^n-f_1$ with different portfolio revision numbers "n" \, and different parameters $\sigma$. Let $S_0=K=100$, $\kappa_0=0.05$.

\begin{table}[h!]
\caption{The hedging error when $\sigma=0.1$}
\label{tab_HedErr1}
\begin{center}
\begin{tabular}{|c||*{6}{c|}}
\hline
$n$ & 20 & 50 & 100 & 200 & 500 & 1000 \\
\hline
$V_1^n-f_1$ & -0.3264 & -0.1479 & -0.0693 & -0.0097 & 0.0026  & 0.0061  \\
 \hline
\end{tabular}
\end{center}
\end{table}

\begin{table}[h!]
\caption{The hedging error when $\sigma=0.9$}
\label{tab_HedErr2}
\begin{center}
\begin{tabular}{|c||*{6}{c|}}
\hline
$n$ & 20 & 50 & 100 & 200 & 500 & 1000 \\
\hline
$V_1^n-f_1$ & -0.7106 & -0.4065 & -0.3307 & -0.1938 & -0.0801  & -0.0213  \\
 \hline
\end{tabular}
\end{center}
\end{table}

An analysis of the numerical results showed that the value of the option increases if the strike price is less than the initial value of the stock. Volatility also affects the value of the option, it increases with increasing volatility, but not significantly. The portfolio revealed an inverse proportion between the number of risky and risk-free assets.
As a result of the experiment, the influence of the number of revisions of the portfolio $ n $ on the value of the option in financial markets with transaction costs was confirmed, it was revealed that with the growth of $ n $ the value of the option also increases. The cost of an option in financial markets without costs does not depend on the number of “revisions”. Also, a numerical experiment showed that in markets with transaction costs, the hedging error decreases with an increase in the number of portfolio revisions. It was also revealed that hedging error is greater with high market volatility.

\section{Appendix}
\subsection{Technical lemmas}\label{sec_dop_lem_small_volatil}

\begin{lemma}\label{shish_lemma1_B_n_small_volatil}
	Let
		\begin{equation}\label{shish_eq_B_n}
		B_n=\kappa_n\sum_{j=1}^n S_{t_j}|H(t_{j},\xi_{t_{j}},S_{t_j})-H(t_{j-1},				
		\xi_{t_{j-1}},S_{t_{j}})|,
		\end{equation}
		$$
			\kappa_n\underset{n\to\infty}{\longrightarrow}0
		$$
	then
		\begin{equation}
		\P-\lim_{n\rightarrow\infty}B_n=0
		\end{equation}
	\end{lemma}
	
	\begin{proof}
 	We can represent
 		\begin{align*}
 		H(t_{j},\xi_{t_{j}},S_{t_j})&=\underbrace{H(t_{j},\xi_{t_{j}},S_{t_j})-H(t_{j-1},	\xi_{t_{j}},S_{t_j})}_{h_j}+\underbrace{H(t_{j-1},\xi_{t_{j}},S_{t_j})-H(t_{j-1},\xi_{t_{j-1}},S_{t_j})}_{h_j^{(1)}}
 		\\
 		&+H(t_{j-1},\xi_{t_{j-1}},S_{t_j}).
 		\end{align*}
Then
 		$$
 		B_n=D_n^{(t)}+D_n^{(x)}=\kappa_n\sum_{j=1}^n d_j^{(t)}+\kappa_n\sum_{j=1}^n d_j^{(x)},
 		$$
 	where
 		$$
 		d_j^{(t)}=S_{t_j} h_j,\;\;\;\;\;\;d_j^{(x)}=S_{t_j}h_j^{(1)}.
 		$$
	It is necessary to show that $\forall\mu>0$
	\begin{equation}\label{shish_main_eq_of_lem_B_n_small_volatil}
	\lim\limits_{n\rightarrow\infty}\P(D_n^{(t)}>\mu)=0\;\;\;\text{and}\;\;\;
	\lim\limits_{n\rightarrow\infty}\P(D_n^{(x)}>\mu)=0.
	\end{equation}

	Recall that
	$$
	H(t,x,y)=\hat{G}_y(t,x,y)=\int\limits_b^{\infty}z\hat{q}(v,z)dz,
	$$
	where $\hat{q}(v,z)$ -- the density of the random density
	$$
	\hat{\eta}_v=\int\limits_0^v\exp\{\hat{\sigma}W_u-\hat{\sigma}^2u/2\}du
	$$  
	and
	$$
	b=\left(\frac{K-x}{y}\right)_+.
	$$
	We introduce the stopping time
		\begin{equation}\label{shish_stop_time}
			\tau_0=\inf\{t>0:\xi_t\geq K\}\wedge 1.
		\end{equation}	
	Clear that always $0<\tau_0\leq 1$ and starting from $\tau_0$, all coefficients $b=0$. To compensate $b$ and $S_t$ we introduce the following sets
	$$
	\Gamma_{\varepsilon,M}=\{\xi_1>K\} \cap \{\tau_0\leq 1-\varepsilon\} \cap \{ M^{-1}\leq \min\limits_{0\leq t\leq 1}S_t\leq \max\limits_{0\leq t\leq 1}S_t\leq M\}
	$$
	$$
	\tilde{\Gamma}_{\delta,M}=\{\xi_1\leq K-\delta\} \cap \{ M^{-1}\leq \min\limits_{0\leq t\leq 1}S_t\leq \max\limits_{0\leq t\leq 1}S_t\leq M\}
	$$ 
	moreover
	\begin{equation}
	\lim\limits_{M\to\infty}\lim\limits_{\varepsilon\to 0}\P(\Gamma_{\varepsilon,M}^c)=0, \;\;\;\;\;
	\lim\limits_{M\to\infty}\lim\limits_{\delta\to 0}\P(\tilde{\Gamma}_{\delta,M}^c)=0.
	\end{equation}		
	
Then we represent probabilities as
 		\begin{align}\label{shish_eq_D_n_t_equal_sum}
		\P (D_n^{(t)}>\mu)&=\P(D_n^{(t)}>\mu,\xi_1>K)+\P (D_n^{(t)}>\mu,\xi_1\leq K)\notag
		\\[3mm]
		&\leq \P(D_n^{(t)}>\mu,\Gamma_{\varepsilon,M})+\P(\Gamma^c_{\varepsilon,M})+
		\P(D_n^{(t)}>\mu,\tilde{\Gamma}_{\delta,M})+\P(\tilde{\Gamma}^c_{\delta,M})
		\end{align}
		and
		\begin{align}\label{shish_eq_D_n_x_equal_sum}
		\P (D_n^{(x)}>\mu)&=\P(D_n^{(t)}>\mu,\xi_1>K)+\P (D_n^{(x)}>\mu,\xi_1\leq K)\notag
		\\[3mm]
		&\leq \P(D_n^{(x)}>\mu,\Gamma_{\varepsilon,M})+\P(\Gamma^c_{\varepsilon,M})+
		\P(D_n^{(x)}>\mu,\tilde{\Gamma}_{\delta,M})+\P(\tilde{\Gamma}^c_{\delta,M})
		\end{align}
	Taking into account Proposition \ref{shish_dop_prop1_for_lemma1_B_n_small_volatil} and Proposition \ref{shish_dop_prop2_for_lemma1_B_n_small_volatil} we obtain the equalities (\ref{shish_main_eq_of_lem_B_n_small_volatil}). 	
 \end{proof}

 	\begin{proposition}\label{shish_dop_prop1_for_lemma1_B_n_small_volatil}
	For fixed $\varepsilon>0$ and $M>0$ 		
 		\begin{equation}
			\lim\limits_{n\rightarrow\infty}\P(D_n^{(t)}>\mu,\Gamma_{\varepsilon,M})=0;
		\end{equation}
		\begin{equation}
			\lim\limits_{n\rightarrow\infty}\P(D_n^{(x)}>\mu,\Gamma_{\varepsilon,M})=0.
		\end{equation}
	Here
	$$
	\Gamma_{\varepsilon,M}=\{\xi_1>K\} \cap \{\tau_0\leq 1-\varepsilon\} \cap \{ M^{-1}\leq \min\limits_{0\leq t\leq 1}S_t\leq \max\limits_{0\leq t\leq 1}S_t\leq M\}
	$$	
	\end{proposition}
 	\begin{proof}
	First we divide  $D_n^{(t)}$ and $D_n^{(x)}$ into two amounts
	\begin{align*}
	D^{(t)}_n&=\kappa_n\sum_{j=1}^{n_1} d^{(t)}_j+\kappa_n\sum_{j=n_1+1}^{n} d^{(t)}_j,
	\\[3mm]
	D^{(x)}_n&=\kappa_n\sum_{j=1}^{n_1} d^{(x)}_j+\kappa_n\sum_{j=n_1+1}^{n} d^{(x)}_j,
	\;\;\;\;n_1=[(1-\varepsilon)n].
	\end{align*}	
	We choose $n_1$ from the condition $t_{n_1}\geq 1-\varepsilon$. In this case, we have 
	$$
	(t_j)_{j=n_1+1}^n>\tau_0 \;\;\;\text{and}\;\;\;(\xi_{t_j})_{j=n_1+1}^n\geq K.
	$$
	Therefore
	$$
	b_{j}=0,\;\;\;j=\overline{n_1+1,n}
	$$
	and since
	\begin{align*}
	H(t,x,y)&=\int\limits_b^{\infty} z\hat{q}(v,z)dz=\int\limits_0^{\infty} z\hat{q}(v,z)dz
	=\E\hat{\eta}_v
	\\[3mm]
	&=\E\left(\int\limits_0^v\exp\{\hat{\sigma}W_u-\hat{\sigma}^2u/2\}du\right)
	=\int\limits_0^v \E\exp\{\hat{\sigma}W_u-\hat{\sigma}^2u/2\}du=v
	\end{align*}
	we obtain
	$$
	d^{(t)}_j=S_{t_j}|H(t_{j},\xi_{t_{j}},S_{t_j})-H(t_{j-1},\xi_{t_{j}},S_{t_j})|\leq M \Delta v_j\leq \frac{M}{n},\;\;\;j=\overline{n_1+1,1}
	$$
	$$
	\kappa_n \sum_{j=n_1+1}^{n} d^{(t)}_j\underset{n\to\infty}{\longrightarrow} 0
	$$
	и
	$$
	d^{(x)}_j=S_{t_j}|H(t_{j-1},\xi_{t_{j}},S_{t_j})-H(t_{j-1},\xi_{t_{j-1}},S_{t_j})|=0,\;\;\;\;j=\overline{n_1+1,1}.
	$$
	Thus, on the set $\Gamma_{\varepsilon,M}$ we have
	$$
	D^{(t)}_n= \kappa_n\sum_{j=1}^{n_1} d^{(t)}_j;\;\;\;D^{(x)}_n= \kappa_n\sum_{j=1}^{n_1} d^{(x)}_j.
	$$
	
	Next we evaluate $d^{(t)}_j$ on the interval $[t_1,t_{n_1}]$, on which all $t_j<1-\varepsilon$.
	\begin{align*}
	d^{(t)}_j&=S_{t_j}\bigg| \int\limits_{b_{t_j}}^{\infty} z\hat{q}(v_j,z)dz - 
	\int\limits_{b_{t_{j-1}}}^{\infty} z\hat{q}(v_{j-1},z)dz\bigg|
	\\[3mm]
	&\leq M \int\limits_{0}^{\infty} z  |\hat{q}(v_j,z)-\hat{q}(v_{j-1},z)|  dz
	\leq M \int\limits_{0}^{\infty} z  \bigg| \int\limits_{v_{j-1}}^{v_j} \hat{q}_v(u,z) du \bigg|  dz	
	\\[3mm]
	&\leq  \int\limits_{0}^{1} z  \bigg| \int\limits_{v_{j-1}}^{v_j} \hat{q}_v(u,z) du \bigg|  dz
	+ \int\limits_{1}^{\infty} z  \bigg| \int\limits_{v_{j-1}}^{v_j} \hat{q}_v(u,z) du \bigg|  dz.
	\end{align*}
	Naking into account Proposition \ref{shish_density_estim_part2}, we can estimate
	$$
		\bigg|\int_{v_{j-1}}^{v_j}\hat{q}_v(u,z)du\bigg|\leq\frac{c}{\varepsilon^4}\int_{v_{j}}^{v_{j-1}}\exp\{-\theta(\ln(z/v))^2\}du,
		$$   
	because we can evaluate $1/v\leq 1$. If $0<z\leq 1$, then we can use a uniform estimate for $\hat{q}_v(v,z)$	
	$$
	\hat{q}_v(v,z)\leq \frac{c}{v^4} 
	\leq \frac{c}{\varepsilon^4}
	$$
	and if $z>1$, we use inequality $\ln(z/v)\geq \ln(z/\varepsilon)$. Then
	$$
	\hat{q}_v(v,z)\leq \frac{c}{\varepsilon^4}
	\exp\left\lbrace-\theta(\ln(z/\varepsilon))^2\right\rbrace.
	$$
	Therefore,
	$$
	d^{(t)}_j\leq \frac{M c}{2 n \varepsilon^4}+
	\frac{M c}{n\varepsilon^4}
	\int\limits_{1}^{\infty} z \exp\left\lbrace -\theta
	(\ln(z/\varepsilon))^2\right\rbrace   dz,
	$$
	and the integral
	$$ I_1=\int\limits_{1}^{\infty} z \exp\left\lbrace -\theta
	(\ln(z/\varepsilon))^2\right\rbrace   dz
	$$ 
	is converge. Then
	$$
	D^{(t)}_n\leq M\kappa_n\sum_{j=1}^{n_1} d^{(t)}_j\leq \kappa_n\sum_{j=1}^{n_1}\frac{M c}{n \varepsilon^4} (1/2+I_1)
	=\kappa_n\frac{M c n_1}{n \varepsilon^4} (1/2+I_1)\xrightarrow[n\rightarrow\infty]{\mathbf{P}}0
	$$
	that is 
	$$
	\lim\limits_{n\rightarrow\infty}\P(D_n^{(t)}>\mu,\Gamma_{\varepsilon,M})=0.
	$$
	
	Let's estimate $d^{(x)}_j$ on the interval $[t_1,t_{n_1}]$. Since 
	$H(t,x,y)=\int\limits_{b}^{\infty}z\hat{q}(v,z)dz$,		
	then
	$$
		H'_x(t,x,y)(t,x,y)=-\frac{b}{y}\hat{q}(v,b)
	$$
	and we can evaluate this derivative using Proposition \ref{shish_density_estim_part1}.
		$$
			|H_{x}(t,x,y)|\leq\frac{bq(v,b)}{y}\leq 
			\frac{K M^2 c}{\varepsilon^2} 
		$$
	We used a uniform estimate for $\hat{q}(v,b)$. Thus the derivative $H_{x}(t,x,y)$ is uniformly bounded. This means that the function $H(t,x,y)$ satisfies the conditions of Lipschitz and moreover
		\begin{align*}
		|H(t_{j-1},\xi_{t_{j}},S_{t_j})-H(t_{j-1},\xi_{t_{j-1}},S_{t_j})|&\leq |H_{x}(t,x,y)|
		|\xi_{t_j}-\xi_{t_{j-1}}|\leq \frac{K M^2 с}{\varepsilon^2}|\xi_{t_j}-\xi_{t_{j-1}}|	
		\\[3mm]
		&\leq\frac{K M^2 с}{\varepsilon^2}\int_{t_{j-1}}^{t_j}S_udu\leq 
		\frac{K M^3 с}{\varepsilon^2} \Delta t_j=\frac{K M^3 с}{n\varepsilon^2}.
		\end{align*}
	Then
		$$
			D^{(x)}_n\leq\kappa_n M\sum_{j=1}^{n_1}
			\frac{K M^3 с}{n\varepsilon^2}
			\leq
			\kappa_n\frac{ K M^4 c n_1}{n\varepsilon^2}
			\xrightarrow[n\rightarrow\infty]{\mathbf{P}}0.
		$$
	that is
	$$
	\lim\limits_{n\rightarrow\infty}\P(D_n^{(x)}>\mu,\Gamma_{\varepsilon,M})=0.
	$$
	
	\end{proof}

 	\begin{proposition}\label{shish_dop_prop2_for_lemma1_B_n_small_volatil}
 		For fixed $\delta>0$ and $M>0$ 		
 		\begin{equation}
			\lim\limits_{n\rightarrow\infty}\P(D_n^{(t)}>\mu,\tilde{\Gamma}_{\delta,M})=0;
		\end{equation}
		\begin{equation}
			\lim\limits_{n\rightarrow\infty}\P(D_n^{(x)}>\mu,\tilde{\Gamma}_{\delta,M})=0.
		\end{equation}
	Here
	$$
	\tilde{\Gamma}_{\delta,M}=\{\xi_1\leq K-\delta\} \cap \{ M^{-1}\leq \min\limits_{0\leq t\leq 1}S_t\leq \max\limits_{0\leq t\leq 1}S_t\leq M\}
	$$ 	
	\end{proposition}
	
 	\begin{proof}
	We can use the following estimate
	$$
	D_n^{(x)}\leq\kappa_n M\sum_{j=1}^{n}|H_{x}(t,x,y)||\xi_{t_j}-\xi_{t_{j-1}}|,
	$$
	since
	$$
	|H_{x}(t,x,y)|\leq\frac{b\hat{q}(v,b)}{y}\leq KM^2\frac{c}{v^2}
			\left(\exp\left\lbrace -\theta(\ln(b/v))^2\right\rbrace\Chi_{\{z>v\}}+\Chi_{\{z\leq v\}}\right).
	$$
	It's obvious that $\xi_{t_j}\leq \xi_1\leq K-\delta$ on the set $\tilde{\Gamma}_{\delta,M}$, so
	$$
	b_*=\frac{\delta}{M}\leq b_{t_j}\leq KM.
	$$
	Similar to the proof of the Proposition \ref{shish_dop_prop2_for_lemma1_B_n_small_volatil} we split $D_n^{(x)}$ into two amounts 
	$$
	D^{(x)}_n=\kappa_n\sum_{j=1}^{n_1} d^{(x)}_j+\kappa_n\sum_{j=n_1+1}^{n} d^{(x)}_j.
	$$
	If $0<v<b_*$, i.e. $1-b_*<t<1$, then we use estimate
	$$
	|H_{x}(t,x,y)|\leq\frac{KM^2 c}{v^2}\exp\left\lbrace -\theta(\ln(b_*/v))^2\right\rbrace.
	$$
	If $b_*<v<1$, i.e. $0\leq t\leq 1-b_*$, then
	$$
	|H_{x}(t,x,y)|\leq\frac{KM^2 c}{b_*^2}.
	$$
	Thus,
		\begin{align*}
			D^{(x)}_n&\leq \frac{\kappa_n c M^3K}{b_*^2}\sum_{j=1}^{n_1}|\xi_{t_j}-\xi_{t_{j-1}}|+
			M^3K\kappa_n\sum_{j=n_1+1}^{n}\frac{c}{v_j^2}
			\exp\{-\theta(\ln(b_*/v_j))^2\}|\xi_{t_j}-\xi_{t_{j-1}}|
			\\[3mm]
			&\leq \frac{c M^4 K n_1\kappa_n}{n b_*^2}+M^4 K\kappa_n
			\sum_{j=n_1+1}^{n}\frac{c}{v_j^2}\exp\{-\theta(\ln(b_*/v_j))^2\}\Delta t_j
		\end{align*}
	Note that
		$$
			\sum_{j=n_1+1}^{n}\frac{c}{v_j^2}\exp\{-\theta(\ln(b_*/v_j))^2\}\Delta t_j					\xrightarrow[n\rightarrow\infty]{}\int\limits_{1-b_*}^1\frac{c}{v^2}\exp\{-\theta(\ln(b_*/v))^2\}dt.
		$$
	Consider
	$$
	J=\int\limits_{1-b_*}^1\frac{c}{v^2}\exp\{-\theta(\ln(b_*/v))^2\}dt.
	$$ 
	Given that $v=1-t$, make a variable change
	$$
	\tilde{a}=\frac{1}{1-t},
	$$ 
	then $dt=d\tilde{a}/\tilde{a}^2$ and we obtain
		\begin{align*}
			J&=\int_{1/b_*}^{\infty}\tilde{a}^2\exp\{-\theta(\ln(b_*\tilde{a}))^2\}d						\tilde{a}
			=\frac{1}{b_*}\int_{1}^{\infty}\exp\{-\theta(\ln\hat{a})^2\}d\hat{a}
			\\[3mm]
			&=\frac{1}{b_*}\int_{0}^{\infty}\exp\{-\theta(\ln y)^2+y\}dy
		\end{align*}
	The last integral $I_2=\frac{1}{b_*}\int_{0}^{\infty}\exp\{-\theta(\ln y)^2+y\}dy$ converges. Thus,
		$$
			D_n^{(x)}\leq \kappa_n c KM^4
			\left(\frac{n_1}{nb_*^2}+I_2\right) 
			\xrightarrow[n\rightarrow\infty]{\P}0
		$$
	and
	$$
	\lim\limits_{n\rightarrow\infty}\P(D_n^{(x)}>\mu,\tilde{\Gamma}_{\delta,M})=0.
	$$
	
	Consider
	$$
	D_n^{(t)}=\kappa_n\sum_{j=1}^{n}d_j^{(t)}.
	$$
		
	Recall that on the set $\tilde{\Gamma}_{\delta,M}$ all  $\xi_{t_j}\leq\xi_1\leq K-\delta$, then
	$$
	b_{t_j}=\frac{K-\xi_{t_j}}{S_{t_j}}\geq b_*,\;\;\;\delta_1=\frac{\delta}{M}.
	$$
	Similarly, we evaluate $d_j^{(t)}$.
		\begin{align*}
		d_j^{(t)}&=S_{t_j}\bigg|\int_{b_{t_j}}^{\infty}z\hat{q}(v_j,z)dz-\int_{b_{t_{j-1}}}^{\infty}
		z\hat{q}(v_{j-1},z)dz\bigg|\leq M \int_{b_*}^{\infty}z \bigg|\int_{v_{j-1}}^{v_j}q_v(u,z)du	\bigg|dz
		\\[3mm]
		&\leq M\int_{b_*}^{\infty}z\left(\int_{v_{j}}^{v_{j-1}}|\hat{q}_v(u,z)|du\right)dz=
		M\int_{v_{j}}^{v_{j-1}}\left(\int_{b_*}^{\infty}z|\hat{q}_v(u,z)|dz\right)du
		\end{align*}
	Next, we use the estimate for the derivative of the density $\hat{q}_v(u,z)$, obtained in Proposition \ref{shish_density_estim_part2},
		\begin{align*}
			\int_{v_{j}}^{v_{j-1}}\left(\int_{b_*}^{\infty}z|\hat{q}_v(u,z)|dz\right)du
			\leq\int_{v_{j}}^{v_{j-1}}\frac{c}{u^4}\left(\int_{b_*}^{\infty}z
			\exp\{-	\theta(\ln(z/u))^2\}dz\right)du
		\end{align*}
	Make the change of variable $y=z/u$
	\begin{align*}
			\int_{v_{j}}^{v_{j-1}}\left(\int_{b_*}^{\infty}z|\hat{q}_v(u,z)|dz\right)du
			&\leq
			\int_{v_{j}}^{v_{j-1}}\frac{c}{u^2}\left(\int_{b_*/u}^{\infty}y
			\exp\{-\theta(\ln y)^2\}dy\right)du
			\\[3mm]
			&=\int_{v_{j}}^{v_{j-1}}\frac{c}{u^2} I(u)du
		\end{align*}
	We consider separately the integral
	$$
	I(u)=\int_{b_*/u}^{\infty}y	\exp\{-\theta(\ln y)^2\}dy.
	$$
  We use the change of variable $x=\ln y$ and obtain
		\begin{align*}
		I(u)=\int_{\ln(b_*/u)}^{\infty}\exp\{-\theta x^2+2x\}dx.
		\end{align*}
Select a parabola so that $c_*=\sup\limits_{x}(\exp\{-\theta x^2/2+2x\})$, then we can write an estimate
\begin{align*}
I(u)&\leq c_*\int_{\ln(b_*/u)}^{\infty}\exp\{-\theta x^2/2\}dx
= c_*\int_{\ln(b_*/u)}^{\infty} \exp\left\lbrace - \theta x^2/4 \right\rbrace 
\exp\left\lbrace - \theta x^2/4 \right\rbrace dx
\\[3mm]
&\leq c_*\exp\left\lbrace - \theta (\ln (b_*/u))^2/4 \right\rbrace \int_{-\infty}^{\infty} \exp\left\lbrace - \theta x^2/4 \right\rbrace dx
\end{align*}		
	If $u\geq b_*$, then
	$$
	I(v)\leq c_*\int_{-\infty}^{+\infty}\exp\left\lbrace - \theta x^2/4 \right\rbrace dx.
	$$
	If $0<u<b_*$, then $\ln(b_*/u)>0$ and we can evaluate the integral as
		$$
			I(u)\leq c_*\exp\left\lbrace - \theta (\ln (b_*/u))^2/4 \right\rbrace \int_{0}^{\infty} \exp\left\lbrace - \theta x^2/4 \right\rbrace dx.
		$$
Then
		\begin{align*}
			d^{(t)}_j&\leq M\int_{v_j}^{v_{j-1}}\frac{c}{u^2} I(u)du
			\leq M\int_{0}^{1}\frac{c}{u^2} I(u)du
			\\[3mm]
			&\leq MJ_1\int_0^{b_*}\exp\left\lbrace - \theta (\ln (b_*/u))^2/4 \right\rbrace
			\frac{c}{u^2}du+MJ_2 \int_{b_*}	^1\frac{c}{u^2}du,
		\end{align*}
	with constants
		$$
			J_1=c_*\int_0^{+\infty}\exp\left\lbrace - \theta x^2/4 \right\rbrace dx;\;\;\;\;\;
			J_2=c_*\int_{-\infty}^{+\infty}\exp\left\lbrace - \theta x^2/4 \right\rbrace dx.
		$$
	It is clear that the integral 
	$$
	I_1=\int_{b_*}^1\frac{c}{u^2}du
	$$ 
	converges. Consider the integral
	$$
	I_2=\int_0^{b_*}\exp\left\lbrace - \theta (\ln (b_*/u))^2/4 \right\rbrace
			\frac{c}{u^2}du .
	$$ 
	Let's make a change $y=b_*/u$, then a change $z=\ln y$, so
		$$
		I_2=\frac{c}{b_*^2}\int_{1}^{\infty} \exp\left\lbrace - \theta (\ln y)^2/4 \right\rbrace
		dy
		=\frac{c}{b_*^2}\int_{1}^{\infty} \exp\left\lbrace - \theta z^2/4+z \right\rbrace
		dz<+\infty
		$$
	Thus, it was shown that $d^{(t)}_j$ is limited by some constant
	$$
	d^{(t)}_j\leq l(M,J_1,J_2,I_1,I_2).
	$$
	Therefore,
		$$
			D^{(t)}_n=\kappa_n\sum_{j=1}^{n} d^{(t)}_j\leq \frac{\kappa_n l(M,J_1,J_2,I_1,I_2) }{n}\xrightarrow[n\rightarrow\infty]{\P}0
		$$
	and we get that
	$$
	\lim\limits_{n\rightarrow\infty}\P(D_n^{(t)}>\mu,\tilde{\Gamma}_{\delta,M})=0.
	$$.

	\end{proof}

	\begin{lemma}\label{shish_lem2_D_n_small_volatil}
		Let	
		$$
			D_n^{(y)}=\kappa_n\sum_{j=1}^n d_j^{(y)},
		$$
		where
		$$
			d_j^{(y)}=S_{t_j}|H(t_{j-1},\xi_{t_{j-1}},S_{t_{j}})-H(t_{j-1},											\xi_{t_{j-1}},S_{t_{j-1}})-H_y(t_{j-1},\xi_{t_{j-1}},S_{t_{j-1}})(S_{t_{j}}-							S_{t_{j-1}})|.
		$$
		Then
		\begin{equation}\label{shish_eq_main_prob_D_n^y_small_volatil}
			\P-\lim_{n\rightarrow\infty}D_n^{(y)}=0.
		\end{equation}
	\end{lemma}
	
\begin{proof}
	
	Represent $d_j^{(y)}$ as
		\begin{align}\label{shish_eq_dj2}	
			d_j^{(y)}&\leq \bigg|S_{t_j}\int\limits_{S_{t_{j-1}}}^{S_{t_j}}
			(H_y(t_{j-1},\xi_{t_{j-1}},u)-H_y(t_{j-1},\xi_{t_{j-1}},S_{t_{j-1}}))du\bigg|\notag\\
			&\leq
			\bigg|S_{t_j}\int\limits_{S_{t_{j-1}}}^{S_{t_j}}\left(\int\limits_{S_{t_{j-1}}}^{u}						H_{yy}(t_{j-1},\xi_{t_{j-1}},a)da\right)du\bigg|.
		\end{align}
	Let's find the derivative $H_{yy}(t,x,y)$. Recall that
		$$
			H(t,x,y)=\int\limits_{b}^{\infty} z\hat{q}(v,z)dz \;\;\;\text{и} \;\;\;
			b=\frac{(K-x)_+}{y}.
		$$		
	If $x\geq K$, then $H(t,x,y)=1-t$ and $H_y(t,x,y)=H_{yy}(t,x,y)=0$. If $x<K$, then
		$$
			H_y(t,x,y)=-b\hat{q}(v,b)b'_y=\frac{b^2}{y}\hat{q}(v,b)=\frac{(K-x)^2}{y^3}
			\hat{q}(v,b)
		$$
	and
		\begin{align*}
			H_{yy}(t,x,y)&=-\frac{3(K-x)^2}{y^4}\hat{q}(v,b)+\frac{(K-x)^2}{y^3}\hat{q}_z(v,b)b'_y\\
			&=-\frac{3b^2}{y^2}\hat{q}(v,b)-\frac{b^3}{y^2}\hat{q}_z(v,b)
		\end{align*}
It is necessary to show that for $\forall \mu>0$	
		$$
			\lim\limits_{n\rightarrow\infty}\P(D_n^{(y)}>\mu)=0.
		$$		
	As before in the proof of Lemma \ref{shish_lemma1_B_n_small_volatil} we introduce the stopping time $\tau_0$ and sets
	$$
	\Gamma_{\varepsilon,M}=\{\xi_1>K\} \cap \{\tau_0\leq 1-\varepsilon\} \cap \{ M^{-1}\leq \min\limits_{0\leq t\leq 1}S_t\leq \max\limits_{0\leq t\leq 1}S_t\leq M\}
	$$
	$$
	\Gamma_{\delta,M}=\{\xi_1\leq K-\delta\} \cap \{ M^{-1}\leq \min\limits_{0\leq t\leq 1}S_t\leq \max\limits_{0\leq t\leq 1}S_t\leq M\},
	$$ 
	with
	\begin{equation}\label{shish_eq_probab_zero}
	\lim\limits_{M\to\infty}\lim\limits_{\varepsilon\to 0}\P(\Gamma_{\varepsilon,M}^c)=0, \;\;\;\;\;
	\lim\limits_{M\to\infty}\lim\limits_{\delta\to 0}\P(\Gamma_{\delta,M}^c)=0.
	\end{equation}
	Represent the probability $\P(D_n^{(y)}>\mu)$ as
		\begin{align*}
		\P (D_n^{(y)}>\mu)&=\P(D_n^{(y)}>\mu,\xi_1>K)+\P (D_n^{(y)}>\mu,\xi_1\leq K)
		\\[3mm]
		&\leq \P(D_n^{(y)}>\mu,\Gamma_{\varepsilon,M})+\P(\Gamma_{\varepsilon,M}^c)+\P(D_n^{(y)}>\mu,\Gamma_{\delta,M})+\P(\Gamma_{\delta,M}^c).
		\end{align*}
	Consider $\P(D_n^{(y)}>\mu,\Gamma_{\varepsilon,M})$. Analogically we split $D_n^{(y)}$ into two sums
		$$
			D_n^{(y)}=\kappa_n\sum_{j=1}^{n_1} d^{(y)}_j+
			\kappa_n\sum_{j=n_1+1}^n d^{(y)}_j,
		$$
	where $n_1=[(1-\varepsilon)n]$, i.e. $t_{n_1}\leq 1-\varepsilon$. Starting from $j=n_1+1$ all moments $t_j>\tau_0$ and $\xi_{t_j}\geq K$. So, $b_{t_j}=0$ and $d^{(y)}_j=0$.	Thus,
		\begin{align*}
			D_n^{(y)}&=\kappa_n\sum_{j=1}^{n_1} d^{(2)}_j
			\\
			&\leq 
			M \kappa_n\sum_{j=1}^{n_1}|H(t_{j-1},\xi_{t_{j-1}},S_{t_{j}})-
			H(t_{j-1},\xi_{t_{j-1}},S_{t_{j-1}})-
			H_y(t_{j-1},\xi_{t_{j-1}},S_{t_{j-1}})(S_{t_{j}}-S_{t_{j-1}})|\\
			&\leq M \kappa_n\sum_{j=1}^{n_1}
			\bigg|\int\limits_{S_{t_{j-1}}}^{S_{t_j}}\left(\int\limits_{S_{t_{j-1}}}^{u} 
			H_{yy}(t_{j-1},\xi_{t_{j-1}},a)da\right)du\bigg|.
		\end{align*}
	Using Proposition  \ref{shish_density_estim_part1}, we can estimate $1/v\leq 1$ and
		\begin{align*}
			|H_{yy}(t,x,y)|&=\bigg|\frac{3b^2}{y^2}\hat{q}(v,b)+
			\frac{b^3}{y^2}\hat{q}_z(v,b)\bigg|
			\\[3mm]
			&\leq \bigg|\frac{3b^2}{y^2}\frac{c}{v^2}\exp\{-\theta(\ln (b/v))^2\}+
			\frac{b^3}{y^2}\frac{c}{v^4}\exp\{-\theta(\ln (b/v))^2\}\bigg|.
		\end{align*}
		On the set $\Gamma_{\varepsilon,M}$ we change the exponent by 1 and $v=1-t$ by $\varepsilon$, also take into account that $b\leq K/y\leq KM$. For some constant $c_*$ we have
		$$
			|H_{yy}(t,x,y)|\leq \frac{c}{y^4\varepsilon^2}+\frac{c}{y^5\varepsilon^4}\leq 
			\frac{c_* M^5}{\varepsilon^4}.
		$$	
	Then
		\begin{align*}
			d^{(y)}_j&\leq M \bigg|\int\limits_{S_{t_{j-1}}}^{S_{t_j}}
			\left(\int\limits_{S_{t_{j-1}}}^{u}H_{yy}(t_{j-1},\xi_{t_{j-1}},a)da\right)du\bigg|\\[3mm]
			&\leq M^6\frac{c_*}{\varepsilon^4}
			\bigg|\int\limits_{S_{t_{j-1}}}^{S_{t_j}} (u-S_{t_{j-1}})du\bigg|=
			\frac{c M^6}{\varepsilon^4}|S_{t_{j}}-S_{t_{j-1}}|^2
		\end{align*}
	Note that
		\begin{align*}
		\E(S_{t_j}-S_{t_{j-1}})^2=\sigma^2\E\int\limits_{t_{j-1}}^{t_j}S_u^2du\leq \sigma^2e^{\sigma^2}(t_j-t_{j-1})=\frac{c}{n},
		\end{align*}
		Since
		$$
		\E S_u^2=\E\exp\{2\sigma W_u-\sigma^2 u\}=e^{\sigma^2 u}\leq e^{\sigma^2}.
		$$		
	Therefore,
		$$
			D^{(y)}_n\leq M \kappa_n\sum_{j=1}^{n_1} \frac{c M^6}{\varepsilon^4}|S_{t_{j}}-S_{t_{j-1}}|^2		
		$$		
	Let $\tilde{c}=c M^7/\varepsilon^4$, then
	\begin{align*}
	\P\left(D_n^{(y)}>\mu,\Gamma_{\varepsilon,M}\right)&\leq\P\left(\kappa_n \tilde{c}\sum_{j=1}^{n_1} |S_{t_{j}}-S_{t_{j-1}}|^2>\mu,\Gamma_{\varepsilon,M}\right)
	\\[3mm]
	&\leq\P\left( \kappa_n \tilde{c} \sum_{j=1}^{n_1}|S_{t_{j}}-S_{t_{j-1}}|^2>\mu\right)
	\end{align*}
	Further, by Chebyshev's inequality, we obtain
	\begin{equation*}
	\P( \kappa_n \tilde{c}\sum_{j=1}^{n_1}|S_{t_{j}}-S_{t_{j-1}}|^2>\mu)
	\leq \frac{ \kappa_n \tilde{c}\sum_{j=1}^{n_1}\E|S_{t_{j}}-S_{t_{j-1}}|^2}{\mu}\xrightarrow[n\rightarrow\infty]{\P}0.
	\end{equation*}
	Thus,
	\begin{equation}\label{shish_eq_D^y_prob_on_Gamma_eps}
			\lim\limits_{n\rightarrow\infty}\P(D_n^{(y)}>\mu,\Gamma_{\varepsilon,M})=0.
		\end{equation}
	
	Consider $\P(D_n^{(y)}>\mu,\Gamma_{\delta,M})$. On the set $\Gamma_{\delta,M}$	
		$$
			b_*=\frac{\delta}{M}\leq b\leq KM,
		$$
	then
		\begin{align*}
			|H_{yy}(t,x,y)|&=\bigg|\frac{3b^2}{y^2}\hat{q}(v,b)+
			\frac{b^3}{y^2}\hat{q}_z(v,b)\bigg|\leq K^3M^5|\hat{q}(v,b)+\hat{q}_z(v,b)|
			\\[3mm]
			&\leq K^3M^5 \bigg|\frac{c}{v^4}\exp\left\lbrace -\theta(\ln (b/v))^2\right\rbrace\bigg|
			\leq K^3M^5 \bigg|\frac{c}{v^4}\exp\left\lbrace -\theta(\ln (b_*/v))^2\right\rbrace\bigg|.
		\end{align*}	
	Let $\tilde{c}=K^3M^5 c$. By making a variable change $z=\ln(b_*/v)$, we obtain
		\begin{align*}			 
			|H_{yy}(t,x,y)|&=\frac{\tilde{c}}{b_*^4}\left(\frac{b_*}{v}\right)^4\exp\left\lbrace -\theta(\ln (b_*/v))^2\right\rbrace
			\leq  \frac{\tilde{c}}{b_*^4}\exp\{-\theta z^2+4z\}
			 \\[3mm]
			 &\leq
			  \frac{\tilde{c}}{b_*^4}\exp\{\sup\limits_{z\in\bbr}(-\theta z^2+4z)\}
			 = \frac{\tilde{c}}{b_*^4}
		\end{align*}			
	Then
		$$
			d_j^{(y)}\leq\frac{\tilde{c}}{b_*^4}|S_{t_j}-S_{t_{j-1}}|^2
		$$
	So, 
		$$
			D_n^{(y)}\leq\kappa_n \frac{\tilde{c}}{b_*^4}\sum_{j=1}^{n} 
			|S_{t_j}-S_{t_{j-1}}|^2
		$$	
		Similarly to the first part of the proof by Chebyshev inequality, we obtain
		$$
			\P(D_n^{(y)}>\mu,\Gamma_{\delta,M})\leq \frac{ \kappa_n c \sum_{j=1}^{n}\E|S_{t_{j}}-S_{t_{j-1}}|^2}{\mu}\xrightarrow[n\rightarrow\infty]{}0.
		$$
	Thus,
		\begin{equation}\label{shish_eq_prob_on_Gamma_delt}
		\lim\limits_{n\rightarrow\infty}\P(D_n^{(y)}>\mu,\Gamma_{\delta,M})=0.
		\end{equation}
	As a result, from the expressions (\ref{shish_eq_probab_zero}),(\ref{shish_eq_D^y_prob_on_Gamma_eps}) and (\ref{shish_eq_prob_on_Gamma_delt}) we obtain (\ref{shish_eq_main_prob_D_n^y_small_volatil}).

	\end{proof}

	\begin{lemma}\label{shish_lem3_converg_small_volatil}
	Let $\beta(t)$ -- continuous  consistent function almost surely. Then
		$$
		\frac{1}{\sqrt{n}}\sum_{j=1}^n\beta(t_{j-1})|S_{t_{j}}-S_{t_{j-1}}|\xrightarrow[n\rightarrow			\infty]{\mathbf{P}}\sqrt{\frac{2}{\pi}}\sigma\int_0^1S_t\beta(t)dt.
		$$
	\end{lemma}
	
	\begin{proof}
	
	All auxiliary constants will be denoted by the letter $c$. We single out the martingale term
\begin{equation}\label{shish_eq_with_martingale_term}
\frac{1}{\sqrt{n}}\sum_{j=1}^n\beta(t_{j-1})|S_{t_{j}}-S_{t_{j-1}}|=\frac{1}{\sqrt{n}}\sum_{j=1}^n\beta(t_{j-1})\mathbf{E}\left(|S_{t_{j}}-S_{t_{j-1}}|\bigg|\mathcal{F}_{t_{j-1}}\right)+\frac{1}{\sqrt{n}}M_n,
\end{equation}
where
$$
M_n=\sum_{j=1}^n\eta_j,
$$
$$
\eta_j=\beta(t_{j-1})|S_{t_{j}}-S_{t_{j-1}}|-\mathbf{E}\left(\beta(t_{j-1})|S_{t_{j}}-S_{t_{j-1}}|\bigg|\mathcal{F}_{t_{j-1}}\right).
$$
In Proposition \ref{shish_prop_dop_small_volatil} we have established that 
$$
\frac{1}{\sqrt{n}}M_n\xrightarrow[n\rightarrow\infty]{\mathbf{P}}0.
$$

Now we consider the first term of the equality (\ref{shish_eq_with_martingale_term}). It's clear that
$$
S_{t_j}-S_{t_{j-1}}=\sigma\int_{t_{j-1}}^{t_j}S_udW_u,
$$ 
so
\begin{align*}
\beta(t_{j-1})\mathbf{E}\left(|S_{t_{j}}-S_{t_{j-1}}|\bigg|\mathcal{F}_{t_{j-1}}\right)&=\sigma\beta(t_{j-1})\mathbf{E}\left(\bigg|\int_{t_{j-1}}^{t_j}S_udW_u\bigg|\bigg|\mathcal{F}_{t_{j-1}}\right)\\
&=\sigma\beta(t_{j-1})\mathbf{E}\left(\bigg|\int_{t_{j-1}}^{t_j}\left[S_{t_{j-1}}+(S_u-S_{t_{j-1}})\right]dW_u\bigg|\bigg|\mathcal{F}_{t_{j-1}}\right)\\
&=\sigma\beta(t_{j-1})\mathbf{E}\left(\bigg|S_{t_{j-1}}(W_{t_{j}}-W_{t_{j-1}})+\int_{t_{j-1}}^{t_j}(S_u-S_{t_{j-1}})dW_u\bigg|\bigg|\mathcal{F}_{t_{j-1}}\right).
\end{align*}
Introduse the notation
\begin{equation}\label{eq_notation_g_j}
g_j=\int_{t_{j-1}}^{t_j}(S_u-S_{t_{j-1}})dW_u.
\end{equation}
Then
\begin{align}\label{shish_eq_3}
\beta(t_{j-1})\mathbf{E}\left(|S_{t_{j}}-S_{t_{j-1}}|\bigg|\mathcal{F}_{t_{j-1}}\right)&=\sigma\beta(t_{j-1})\mathbf{E}\left(\bigg|S_{t_{j-1}}\Delta W_{t_j}+g_j\bigg|\bigg|\mathcal{F}_{t_{j-1}}\right)\notag\\
&=\sigma\beta(t_{j-1})\mathbf{E}\left(\bigg|S_{t_{j-1}}\Delta W_{t_j}\bigg|\bigg|\mathcal{F}_{t_{j-1}}\right)+\sigma\beta(t_{j-1})\mathbf{E}\left(\nu_j\bigg|\mathcal{F}_{t_{j-1}}\right)\notag\\
&=\sigma\sqrt{\frac{2}{\pi}}\beta(t_{j-1})|S_{t_{j-1}}|\frac{1}{\sqrt{n}}+\sigma\beta(t_{j-1})\mathbf{E}\left(\nu_j\bigg|\mathcal{F}_{t_{j-1}}\right),
\end{align}
where
$$
\nu_j=\bigg|S_{t_{j-1}}\Delta W_{t_j}+g_j\bigg|-\bigg|S_{t_{j-1}}\Delta W_{t_j}\bigg|.
$$
By module property $\big||a|-|b|\big|\leq |a-b|$ we have
\begin{equation}\label{eq_est_nu_by_g_j}
|\nu_j|=\bigg|\bigg|S_{t_{j-1}}\Delta W_{t_j}+g_j\bigg|-\bigg|S_{t_{j-1}}\Delta W_{t_j}\bigg|\bigg|\leq |g_j|.
\end{equation}
By the equality (\ref{eq_notation_g_j}) and Jensen's inequalities we get
\begin{align*}
\mathbf{E}\left(|g_j|\big|\mathcal{F}_{t_{j-1}}\right)&=\mathbf{E}\left(\bigg|\int_{t_{j-1}}^{t_j}(S_u-S_{t_{j-1}})dW_u\bigg|\bigg|\mathcal{F}_{t_{j-1}}\right)=\sqrt{\left(\mathbf{E}\left(\bigg|\int_{t_{j-1}}^{t_j}(S_u-S_{t_{j-1}})dW_u\bigg|\bigg|\mathcal{F}_{t_{j-1}}\right)\right)^2}
\\[2mm]
&\leq \sqrt{\mathbf{E}\left(\left(\int_{t_{j-1}}^{t_j}(S_u-S_{t_{j-1}})dW_u\right)^2\bigg|\mathcal{F}_{t_{j-1}}\right)}=\sqrt{\mathbf{E}\left(\int_{t_{j-1}}^{t_j}(S_u-S_{t_{j-1}})^2du\bigg|\mathcal{F}_{t_{j-1}}\right)}
\\[2mm]
&=\sqrt{\int_{t_{j-1}}^{t_j}\mathbf{E}\left((S_u-S_{t_{j-1}})^2\bigg|\mathcal{F}_{t_{j-1}}\right)du}.
\end{align*}
For the integrand we can write the estimate
\begin{equation}\label{estimate integrand}
\mathbf{E}\left((S_u-S_{t_{j-1}})^2\bigg|\mathcal{F}_{t_{j-1}}\right)\leq c^2S^2_{t_{j-1}}(u-t_{j-1})
\end{equation}
This estimate is valid by virtue of the following reasoning.
\begin{align}\label{eq_int_exp_S_v_kvadr}
\mathbf{E}\left((S_u-S_{t_{j-1}})^2\bigg|\mathcal{F}_{t_{j-1}}\right)&=\mathbf{E}\left(\left(\sigma\int_{t_{j-1}}^uS_vdW_v\right)^2\bigg|\mathcal{F}_{t_{j-1}}\right)=\sigma^2\mathbf{E}\left(\int_{t_{j-1}}^uS^2_vdv\bigg|\mathcal{F}_{t_{j-1}}\right)\notag
\\[2mm]
&=\sigma^2\int_{t_{j-1}}^u\mathbf{E}\left(S^2_v\bigg|\mathcal{F}_{t_{j-1}}\right)dv.
\end{align}
By definition of a risky asset in this model, we have
$$
S_v=S_0\exp\{\sigma W_v-v\sigma^2/2\}.
$$
Obviously, the following equality holds.
$$
S^2_v=S_{t_{j-1}}\exp\{2\sigma (W_v-W_{t_{j-1}})-\sigma^2(v-t_{j-1})\}.
$$
Then
\begin{align*}
\mathbf{E}\left(S^2_v\bigg|\mathcal{F}_{t_{j-1}}\right)&=S^2_{t_{j-1}}\mathbf{E}\left(\exp\{2\sigma (W_v-W_{t_{j-1}})-\sigma^2(v-t_{j-1})\}\bigg|\mathcal{F}_{t_{j-1}}\right)
\\[2mm]
&=S^2_{t_{j-1}}\mathbf{E}\left(\exp\{2\sigma\sqrt{v-t_{j-1}}\eta-\sigma^2(v-t_{j-1})\}\right)=S^2_{t_{j-1}}\exp\{\sigma^2(v-t_{j-1})\}
\\[2mm]
&\leq S^2_{t_{j-1}}\exp\{\sigma(t_{j}-t_{j-1})\}\leq S^2_{t_{j-1}}e,
\end{align*}
where $\eta\sim N(0,1)$. Substituting the last inequality into(\ref{eq_int_exp_S_v_kvadr}), we get an estimate (\ref{estimate integrand}). Then
\begin{equation}\label{eq_est_expir_g_j}
\mathbf{E}\left(|g_j|\big|\mathcal{F}_{t_{j-1}}\right)\leq c\sqrt{S^2_{t_{j-1}}\int_{t_{j-1}}^{t_j}(u-t_{j-1})du}=\frac{c}{n}S_{t_{j-1}}.
\end{equation}
Taking into account the equality (\ref{shish_eq_3}), rewrite (\ref{shish_eq_with_martingale_term}) without martingale term
$$
\frac{1}{\sqrt{n}}\sum_{j=1}^n\beta(t_{j-1})|S_{t_{j}}-S_{t_{j-1}}|=\sigma\sqrt{\frac{2}{\pi}}\sum_{j=1}^n\beta(t_{j-1})\frac{1}{n}S_{t_{j-1}}+\frac{\sigma}{\sqrt{n}}\sum_{j=1}^n\beta(t_{j-1})\mathbf{E}\left(|\nu_j|\bigg|\mathcal{F}_{t_{j-1}}\right).
$$
We note here that due to inequalities \eqref{eq_est_nu_by_g_j} and \eqref{eq_est_expir_g_j}
\begin{align*}
\frac{\sigma}{\sqrt{n}}\sum_{j=1}^n\beta(t_{j-1})\mathbf{E}\left(|\nu_j|\bigg|\mathcal{F}_{t_{j-1}}\right)\leq\frac{\sigma}{\sqrt{n}}\sum_{j=1}^n|\beta(t_{j-1})|\mathbf{E}\left(|g_j|\bigg|\mathcal{F}_{t_{j-1}}\right)\leq\frac{c\sigma}{n\sqrt{n}}\sum_{j=1}^n|\beta(t_{j-1})|S_{t_{j-1}}.
\end{align*}
Since
$$
\frac{1}{n}\sum_{j=1}^n|\beta(t_{j-1})|S_{t_{j-1}}\xrightarrow[n\rightarrow\infty]{}\int\limits_0^1|\beta(t)|S_tdt \text{\;\;\;\;\;п.н.},
$$
then
$$
\frac{c\sigma}{\sqrt{n}}\int\limits_0^1|\beta(t)|S_tdt \xrightarrow[n\rightarrow\infty]{\mathbf{P}}0.
$$
Hence, Следовательно, in respect Proposition \ref{shish_prop_dop_small_volatil}, we obtain
$$
\frac{1}{\sqrt{n}}\sum_{j=1}^n\beta(t_{j-1})|S_{t_{j}}-S_{t_{j-1}}|\xrightarrow[n\rightarrow\infty]{\mathbf{P}}\sigma\sqrt{\frac{2}{\pi}}\int_0^1\beta(t)S_tdt.
$$
The lemma is proved.
\end{proof}

	\begin{proposition}\label{shish_prop_dop_small_volatil}
	Let $	M_n=\sum_{j=1}^n\eta_j$, where
	$$
	\eta_j=\beta(t_{j-1})|S_{t_{j}}-S_{t_{j-1}}|-\mathbf{E}	\left(\beta(t_{j-1})|S_{t_{j}}-S_{t_{j-1}}|\bigg|\mathcal{F}_{t_{j-1}}\right).
	$$ 
	For an arbitrary continuous consistent function $\beta(t), 0\leq t\leq 1$ 
		\begin{equation}\label{shish_eq_main_converg}
			\frac{1}{\sqrt{n}}M_n\xrightarrow[n\rightarrow\infty]{\mathbf{P}}0.
		\end{equation}
	\end{proposition}
	
	\begin{proof}
We consider two cases. In the first case, suppose that for some constant $L$
$$
\sup_{0\leq t\leq 1}|\beta(t)|\leq L.
$$
By virtue of the martingality property, we obtain
\begin{equation}\label{shish_eq_2}
\E(M_n)^2=\mathbf{E}\sum_{j=1}^n\eta^2_j=\sum_{j=1}^n\mathbf{E}\mathbf{E}(\eta^2_j|\mathcal{F}_{t_{j-1}})=\mathbf{E}\sum_{j=1}^n\mathbf{E}(\eta^2_j|\mathcal{F}_{t_{j-1}}).
\end{equation}
Introduce the notation $\eta_j=\beta(t_{j-1})\varsigma_j$, where
$$
\varsigma_j=|S_{t_{j}}-S_{t_{j-1}}|-\mathbf{E}\left(|S_{t_{j}}-S_{t_{j-1}}|\big|\mathcal{F}_{t_{j-1}}\right).
$$
Then
$$
\mathbf{E}(\eta^2_j|\mathcal{F}_{t_{j-1}})=\beta^2(t_{j-1})\mathbf{E}(\varsigma^2_j|\mathcal{F}_{t_{j-1}})\leq L^2 \mathbf{E}(|S_{t_{j}}-S_{t_{j-1}}|^2\big|\mathcal{F}_{t_{j-1}})\leq L^2 S^2_{t_{j-1}}\Delta\frac{2}{\pi}.
$$
Then we can evaluate the equality (\ref{shish_eq_2}) 
$$
\mathbf{E}(M_n)^2\leq L^2\mathbf{E}\sum_{j=1}^nS^2_{t_{j-1}}\Delta\frac{2}{\pi} =L^2\frac{2}{\pi n}\sum_{j=1}^n\mathbf{E}S^2_{t_{j-1}}\leq c^2.
$$
Therefore, we obtain
\begin{equation}\label{shish_eq_5}
\mathbf{E}\frac{1}{\sqrt{n}}|M_n|=\frac{1}{\sqrt{n}}\mathbf{E}|M_n|\leq \frac{1}{\sqrt{n}}(\mathbf{E}M_n^2)^{1/2}\leq \frac{c}{\sqrt{n}}\xrightarrow[n\rightarrow\infty]{\mathbf{P}}0
\end{equation}
In the second case, we assume that $\beta(t)$ -- continuous function. We introduce the stopping time
$$
\tau_L=\inf\{t\geq 0: |\beta(t)|\geq L\}\wedge 1.
$$ 
For continuous function $\beta(t)$ the following equality holds
$$
\mathbf{P}(\max_{0\leq t\leq 1}|\beta(t)|<\infty)=1.
$$
It means that $\beta(t)$ is limited and therefore,
\begin{equation}\label{eq_dop_zvezda}
\mathbf{P}(\tau_L<1)\xrightarrow[L\rightarrow\infty]{}0.
\end{equation}
To prove convergence (\ref{shish_eq_main_converg}), it must be shown that the next probability is zero.
\begin{align*}
\mathbf{P}\left(\frac{1}{\sqrt{n}}|M_n|>\delta\right)=\mathbf{P}\left(\frac{1}{\sqrt{n}}|M_n|>\delta,\tau_L=1\right)+\mathbf{P}\left(\frac{1}{\sqrt{n}}|M_n|>\delta,\tau_L<1\right).
\end{align*}
On the set $\{\tau_L=1\}$ 
$$
\beta(t)=\tilde{\beta}_L(t),
$$
where
$$
\tilde{\beta}_L(t)=\beta(t)\Chi_{\{|\beta(t)|\leq L\}}
$$
and
$$
M_n=\tilde{M}_n^{(L)}=\sum_{j=1}^n\tilde{\beta}_L(t_{j-1})\varsigma_j.
$$
Thus, we found ourselves in the conditions of the first case considered above. Therefore,
\begin{align*}
\mathbf{P}\left(\frac{1}{\sqrt{n}}|M_n|>\delta,\tau_L=1\right)&=\mathbf{P}\left(\frac{1}{\sqrt{n}}|\tilde{M}_n^{(L)}|>\delta,\tau_L=1\right)
\\
&\leq \mathbf{P}\left(\frac{1}{\sqrt{n}}|\tilde{M}_n^{(L)}|>\delta\right)\leq \frac{1}{\delta \sqrt{n}}\mathbf{E}|\tilde{M}_n^{(L)}|.
\end{align*}
By (\ref{shish_eq_5}), last expectation tends to zero. Then $\forall L>1$
\begin{align*}
\lim_{n\rightarrow\infty}\mathbf{P}\left(\frac{1}{\sqrt{n}}|M_n|>\delta\right)&=\lim_{n\rightarrow\infty}\left(\mathbf{P}\left(\frac{1}{\sqrt{n}}|\tilde{M}_n^{(L)}|>\delta\right)+\mathbf{P}(\tau_L<1)\right)
\\[2mm]
&\leq \mathbf{P}(\tau_L<1).
\end{align*}
Pass to the limit by $L\to\infty$, with \eqref{eq_dop_zvezda} we obtain \eqref{shish_eq_main_converg}.

\end{proof}

\subsection{Proof of the density properties}
\begin{proof} (Proposition \ref{shish_density_estim_part1})
	
	We need to look at the asymptotic behavior $q(v,z)$ when $v\rightarrow 0$ and $z>0$ is fixed. 
	To obtain an upper estimate for the density $q(v,z)$, it is necessary to estimate the function 			$K(v,a)$ from below. 
		\[
			K(v,a)=\sigma\int_0^v u \exp\{\sigma W_u-\sigma uW_1-\sigma^2 u/2+\sigma ua(v,z)\}du=
			\sigma\int_0^v u \exp\{\sigma W_u+\gamma u\}du,
		\]
	where
		\[
			\gamma=\sigma a(v,z)-\sigma W_1-\sigma^2/2.
		\]		
	Next we make the change of variables $s=u\sigma^2$ to use the scale invariance property of Wiener process. 
	Then
		\[
			\frac{1}{\sigma^3}\int_0^{\sigma^2 v} s\exp\{\sigma W_{s/\sigma^2}+\gamma s/\sigma^2\}ds=
			\frac{1}{\sigma^3}\int_0^{\sigma^2 v} s\exp\{\bar{W}_{s}+\gamma s/\sigma^2\}ds,
		\]	
	here $\bar{W}_{s}=\sigma W_{s/\sigma^2}$ is also the Wiener process. Further we suppose that $\sigma\geq 1$ and estimate $K(v,a)$ as follow
		\begin{align*}
			K(v,a)&\geq\frac{1}{\sigma^3}\int_0^{v_*} s\exp\{\bar{W}_{s}-W_1s/\sigma-s/2+
			s a(v,z)/\sigma\}ds\\
			&\geq\frac{1}{\sigma^3}\exp\{-\max_{0\leq s\leq 1} |\bar{W}_{s}|-|W_1|-|a(v,z)|\}v_*^2
		\end{align*}
	where $v_*=\min (\sigma^2 v,1)$. Let $\beta_1=\max\limits_{0\leq s\leq 1}|\bar{W}_{s}|+|W_1|$
	then
	$$
		K(v,a)\geq\frac{v_*^2}{\sigma^3}\exp\{-\beta_1-|a(v,z)|\}.
	$$
	Substituting this estimate in $q(v,z)$, we have
		\begin{align}\label{shish_prop_dens_eq_1_part1}
			q(v,z)&\leq\frac{c\sigma^3}{v_*^2}\E\exp\{\beta_1-a^2(v,z)/4\},
		\end{align}
	where $c=e/\sqrt{2\pi}$ and $\exp\{\sup_{a}(|a|-a^2/4)\}=e$. 
	Next obtain the lower bound for $a(v,z)$.The function $a(v,z)$ is specified implicitly as follows
		\[
			z=\int_0^v\exp\{\sigma W_u-\sigma uW_1-\sigma^2 u/2+\sigma ua(v,z)\}du, \;\;\; v=1-t.	
		\]
	Consider 
	\begin{align*}
	\beta_*^1=\max\limits_{0\leq u\leq v}|W_u-uW_1|&\leq \max\limits_{0\leq u\leq v}|W_u|+v|W_1|
	\leq\sqrt{v}(\max\limits_{0\leq u\leq v}|W_u|/\sqrt{v}+\sqrt{v} |W_1|)\\
	&\leq\sqrt{v}(\max\limits_{0\leq u\leq v}|W_u|/\sqrt{v}+|W_1|)=\sqrt{v}\beta_*,
	\end{align*} 
	$$
			\beta_*=\max\limits_{0\leq u\leq v}|W_u|/\sqrt{v}+|W_1|\;\;\;\text{and}\;\;\;
			\E e^{N\beta_*}<\infty.
	$$
	Therefore for $z$ the following estimate is hold
		\[
			z\leq v\exp\{\sigma\sqrt{v}\beta_*\}\exp\{\sigma v|a(v,z)|\},
		\]
	Clearly that 
	$$
	\ln\left(\frac{z}{v}\right)\leq \sigma\sqrt{v}\beta_*+\sigma v |a(v,z)|
	$$ 
	from which it follows that
		\begin{equation}\label{estimate_of a_part1}
			|a(v,z)|\geq \frac{1}{\sigma v}\left(\ln(z/v)-\sigma\sqrt{v}\beta_*\right).
		\end{equation}
	Clearly that $\ln (z / v)$ is large for $v\rightarrow 0$ , consider the expectation in (\ref{shish_prop_dens_eq_1_part1})
	 	\begin{align*}\label{ineq with indicators_part1}
	 		\E\exp\{\beta_1-a^2(v,z)/4\}&=\E\exp\{\beta_1-a^2(v,z)/4\}
	 		(\Chi_{\{\beta_*\leq L\}}+\Chi_{\{\beta_*> L\}})\\
	 		&\leq \E\exp\{\beta_1-a^2(v,z)/4\}\Chi_{\{\beta_*\leq L\}}+
	 		\E\exp\{\beta_1\}\Chi_{\{\beta_*> L\}}\\
	 		&\leq \E e^{\beta_1}\exp\{-a^2(v,z)/4\}\Chi_{\{\beta_*\leq L\}}+
	 		(\E e^{2\beta_1})^{1/2}(\P(\beta_*> L))^{1/2}
	 	\end{align*}
	Let $c_1=\max (\E e^{\beta_1},(\E e^{2\beta_1})^{1/2})$. Using Markov's inequality we obtain
		\[
			\P(\beta_*> L)=\P(e^{\delta_*\beta^2_*}>e^{\delta_* L^2})\leq 
			\exp\{-\delta_* L^2\}\E\exp\{\delta_*\beta^2_*\}=c_2^2\exp\{-\delta_* L^2\},
		\]
	Consider 
	$$
	c_2^2=\E\exp\{\delta_*\beta^2_*\}=\sum\limits_{m=0}^{\infty}\frac{\delta_8^m}{m!}\E\beta_*^{2m}.
	$$
	and besides
	\begin{align*}
	\E\beta_*^{2m} &\leq 2^{2m}\E\left( |W_1|^{2m}+(\max\limits_{0\leq u\leq v}|W_u|/\sqrt{v})^{2m} \right)=2^{2m}\left((2m-1)!!+\frac{1}{v^m}\E\max\limits_{0\leq u\leq v}|W_u|^{2m}\right)
	\\[3mm]
	&\leq 2^{2m}\left(2^m m!+\frac{1}{v^m}\left(\frac{2m}{2m-1}\right)^{2m}\E|W_v|^{2m}\right)
	\\[3mm]
	&\leq c 2^{3m} m!+\frac{2^m (2m-1)!! v^m}{v^m}\leq c 2^{3m} m!.
	\end{align*}
	Hence
	$$
	\E\exp\{\delta_*\beta^2_*\} \leq 1+\sum\limits_{m=0}^{\infty} c \delta_*^m 2^{3m}
	$$
	and this series will converge if we choose $\delta_*<1/8$.
	
	Then the expectation take the form
	\begin{align*}
	 		\E\exp\{\beta_1-a^2(v,z)/4\}\leq c_1(\exp\{-a^2(v,z)/4\}\Chi_{\{\beta_*\leq L\}}+
	 		c_2\exp\{-\delta_* L^2/2\})
	 \end{align*}
	If $\beta_*\leq L$ then inequality 	(\ref{estimate_of a_part1}) will take the form
		\[
			|a(v,z)|\geq \frac{1}{\sigma v}\left(\ln(z/v)-\sigma\sqrt{v}\beta_*\right)
			\geq \frac{1}{\sigma v}\left(\ln(z/v)-\sigma\sqrt{v} L\right).
		\]
	The constant $L$ must be chosen so that $\ln(z/v)-\sigma L$>0. Let
		\[
			L=\frac{1}{2\sigma\sqrt{v}}\ln\left( \frac{z}{v} \right).
		\]
	Then 
	$$
	|a(v,z)|>\frac{1}{2\sigma v}\ln(z/v)
	$$ 
	and 
		\begin{align*}
		 \E\exp\{\beta_1-a^2(v,z)/4\}&\leq c_1\left(\exp\left\lbrace -\frac{1}{16\sigma^2 v^2}(\ln(z/v))^2\right\rbrace+c_2\exp\left\lbrace -\frac{\delta_*}{8\sigma^2 v} (\ln(z/v))^2\right\rbrace\right)
		 \\
		 &\leq c_1(1+c_2)\exp\left\lbrace -\frac{\delta_*}{8\sigma^2 v} (\ln(z/v))^2\right\rbrace .
		\end{align*}
	Thus for the constants $c_*=c(1+c_2)c_1$ and $\kappa=\delta_*/8$ which not depend on $\sigma$ we have the following estimate for the density
	$$
		q(v,z)\leq \frac{c_*\sigma^3}{v_*^2}\exp\left\lbrace -\frac{\kappa}{\sigma^2 v} (\ln(z/v))^2\right\rbrace \;\;\;\forall z>v.
	$$

	Next, consider the derivative  of $q(v,z)$ w.r.t. $z$. Let
		\[
			L(v,a(v,z))=\frac{\varphi_{0,1}(a)}{K(v,a)}
		\]
	First of all, we prove that
		$$
		q_z(v,z)=\frac{\partial}{\partial z}\E(L(v,a(v,z)))=\E\frac{\partial}{\partial z}L(v,a(v,z)).
		$$
	Let 
		$$
			\tilde{L}(v,z)=L(v,a(v,z))
		$$
	and
		$$
		\xi_{\Delta}(z)=\frac{\tilde{L}(v,z+\Delta)-\tilde{L}(v,z)}{\Delta}.
		$$
	Then
		$$
		q_z(v,z)=\frac{q(v,z+\Delta)-q(v,z)}{\Delta}=\E\xi_{\Delta}(z).
		$$
	Be the derivative definition we obtain 
		$$
		\xi_{\Delta}(z)\xrightarrow[\Delta\rightarrow 0]{}\frac{\partial}{\partial z}\tilde{L}(v,z).
		$$
	Also we can write $\forall\Delta>0$
	\begin{align*}
	|\xi_{\Delta}(z)|&=\big|\frac{1}{\Delta}\int_{z}^{z+\Delta}\frac{\partial}{\partial u}
	\tilde{L}(v,u)du\big|
	\leq\sup\limits_{z\in\bbr^+}\sup\limits_{0\leq v\leq 1}
	\big|\frac{\partial}{\partial z}\tilde{L}(v,z)\big|<\infty
	\end{align*}
	Hence we can use Lebesgue's theorem
	$$
	\lim\limits_{\Delta\rightarrow 0}\E\xi_{\Delta}(z)=
	\E\lim\limits_{\Delta\rightarrow 0}\xi_{\Delta}=\E \frac{\partial}{\partial z}\tilde{L}(v,z).
	$$
	Thus we obtain that
		\[
			q_z(v,z)=\E\left(\frac{\partial L(v,a(v,z))}{\partial z}\right)=\E(L'_a(v,a(v,z))a'_z.		
		\]
	It's clear that
	\begin{align*}
	L'_a(v,a)&=\frac{\varphi'_{0,1}(a)K(v,a)-\varphi_{0,1}(a)K'_a(v,a)}{K^2(v,a)}
	\\	
	K'_a(v,a)&=\sigma^2\int_0^{v}u^2\exp\left\lbrace\sigma W_u+\sigma uW_1+\sigma ua-\sigma^2u/2\right\rbrace du\leq \sigma K(v,a) 
	\\
	a'_z&=\frac{1}{K(v,a)}.
	\end{align*}
	Therefore we obtain that
	$$
	|q'_z(v,z)|\leq\E\left(\frac{|\varphi'_{0,1}(a)|+\sigma\varphi_{0,1}(a)}{K^2(v,a)}\right)
	\leq (1+\sigma)\E\left(\frac{|\varphi'_{0,1}(a)|+\varphi_{0,1}(a)}{K^2(v,a)}\right).
	$$
	Since
	$$
		K^2(v,a)\geq \frac{v_*^4}{\sigma^6}\exp\{-2\beta_1-2|a|\}
	$$
	then
	$$
		|q'_z(v,z)|\leq\frac{(1+\sigma)\sigma^6}{v_*^4\sqrt{2\pi}}
		\E\exp\{2\beta_1+2|a|-a^2(v,z)/2\}(1+|a|)
	$$
	We can write the follows
		\begin{align*}
			(1+|a|)\exp\{2|a|-a^2/2\}&=(1+|a|)\exp\{-a^2/4\}\exp\{a^2/8+2|a|-a^2/8\}
			\\
			&\leq e^8\exp\{-a^2/4\}(1+|a|)\exp\{-a^2/8\}
			\\
			&\leq c\exp\{-a^2/4\},\;\;\; c=e^8\sup\limits_{x}(1+x)e^{-x^2/8}.
		\end{align*}
	Hence
	\begin{equation*}
			|q'_z(v,z)|\leq\frac{c^1\sigma^7}{v_*^4}\E\exp\{2\beta_1-a^2(v,z)/4\},\;\; 
			c^1=2c/\sqrt{2\pi}.
	\end{equation*}
	Thus we have obtained the estimate of $q'_z(v,z)$ similar to the estimate of $q(v,z)$. Taking into account that in this case the constant $c_1=\max (\E e^{2\beta_1},(\E e^{4\beta_1})^{1/2})$ also carries no information about $\sigma$ we can analogically write the estimate for $c_*^1=c^1(1+c_2)c_1$ and $\kappa=\delta/8$
	$$
		|q'_z(v,z)|\leq\frac{c_*^1\sigma^7}{v_*^4}\exp\left\lbrace -\frac{\kappa}{\sigma^2 v} (\ln(z/v))^2\right\rbrace \;\;\;\forall z>v.
	$$
	Let $\tilde{c}=\max (c_*,c_*^1)$ then we obtain the desired estimates.		
	\end{proof}

	\begin{proof} (Proposition \ref{shish_density_estim_part2})
	
	We need to look at the asymptotic behavior $q_v(v,z)$ when $v\rightarrow 0$ and $z>0$ is fixed. 
	Let
		\[
			L(v,a(t,z))=\frac{\varphi_{0,1}(a)}{K(v,a)}
		\]
	First of all, we prove that
		$$
		q_v(v,z)=\frac{\partial}{\partial v}\E(L(v,a(v,z)))=\E\frac{\partial}{\partial v}L(v,a(v,z)).
		$$
	Let 
		$$
			\tilde{L}(v,z)=L(v,a(v,z))
		$$
	and
		$$
		\xi_{\Delta}(z)=\frac{\tilde{L}(v+\Delta,z)-\tilde{L}(v,z)}{\Delta}.
		$$
	Then
		$$
		q_v(v,z)=\frac{q(v+\Delta,z)-q(v,z)}{\Delta}=\E\xi_{\Delta}(z).
		$$
	Be the derivative definition we obtain 
		$$
		\xi_{\Delta}(z)\xrightarrow[\Delta\rightarrow 0]{}\frac{\partial}{\partial v}\tilde{L}(v,z).
		$$
	Moreover $\forall\Delta>0$
	\begin{align*}
	|\xi_{\Delta}(z)|&=\big|\frac{1}{\Delta}\int_{v}^{v+\Delta}\frac{\partial}{\partial u}
	\tilde{L}(u,z)du\big|
	\leq\sup\limits_{z\in\bbr^+}\sup\limits_{0\leq v\leq 1}
	\big|\frac{\partial}{\partial v}\tilde{L}(v,z)\big|:=\xi^*(z)
	\end{align*}
	and
	$$
	\E\xi^*(z)<\infty.
	$$
	Hence we can use Lebesgue's theorem
	$$
	\lim\limits_{\Delta\rightarrow 0}\E\xi_{\Delta}(z)=
	\E\lim\limits_{\Delta\rightarrow 0}\xi_{\Delta}=\E \frac{\partial}{\partial v}\tilde{L}(v,z).
	$$
	Thus we obtain that
		\[
			q_v(v,z)=\E\left(\frac{\partial L(v,a(v,z))}{\partial v}\right).		
		\]
	
	Next we need to calculate
	$$
	\frac{\partial L(v,a(v,z))}{\partial v}=L'_v(v,a)+L'_a(v,a)a'_v.
	$$
	Introduce the notation	
		$$
			F(v,a)=\int_0^v\exp\{\sigma W_u-\sigma uW_1-\sigma^2 u/2+\sigma ua\}du
		$$
	and
		$$
			P(v,a)=\exp\{\sigma W_v-\sigma vW_1-\sigma^2 v/2+\sigma va\}.
		$$
	Find $a'_v$ from the equality $z=F(v,a(v,z))$. Differentiating by $v$ we obtain
	$$
	0=F_v(v,a)+F_a(v,a)a'_v
	$$
	hence
	$$
	a'_v=-\frac{F_v(v,a)}{F_a(v,a)}=-\frac{P(v,a)}{K(v,a)}
	$$
	It's clear that
	\begin{align*}
	L'_a(v,a)&=\frac{\varphi'_{0,1}(a)K(v,a)-\varphi_{0,1}(a)K'_a(v,a)}{K^2(v,a)}
	\\	
	L'_v(v,a)&=-\frac{\varphi_{0,1}(a)K'_v(v,a)}{K^2(v,a)},
	\end{align*}
	at that $K'_v(v,a)=\sigma vP(v,a)$.
	
	Then
	\begin{align*}
	q'_v(v,z)&=\E\left[-\frac{\varphi_{0,1}(a)K'_v(v,a)}{K^2(v,a)}-\frac{P(v,a)}{K(v,a)}
	\left( \frac{\varphi'_{0,1}(a)K(v,a)-\varphi_{0,1}(a)K'_a(v,a)}{K^2(v,a)}\right)\right]
	\\[10mm]
	&=\E\left[\frac{-\varphi_{0,1}(a)\sigma v P(v,a)K(v,a)-P(v,a)\varphi_{0,1}(a)K(v,a)+
	P(v,a)\varphi_{0,1}(a)K'_a(v,a)}{K^3(v,a)}\right]
	\\[10mm]
	&=\E\left[P(v,a)\varphi_{0,1}(a)\left(\frac{-\sigma v K(v,a)+a K(v,a)+K'_a(v,a)}{K^3(v,a)}\right)\right]
	\end{align*}
	We can estimate
	\begin{align*}
	|q_v(v,z)| & \leq \E\left( (1+|a|)\varphi_{0,1}(a)P(v,a)\frac{(\sigma v+1)K(v,a)+|K'_a(v,a)|}{K^3(v,a)}\right)
	\end{align*}
	Seeing that $K(v,a)=\sigma\int\limits_0^v u P(u,a)du$ and
	$$
		K'_a(v,a)=\sigma^2\int\limits_0^v u^2 P(u,a)du \leq \sigma v K(v,a)
	$$
	we have
	$$
		|q_v(v,z)| \leq \E\left( (1+|a|)\varphi_{0,1}(a)P(v,a)\frac{(\sigma v+1)}{K^2(v,a)}\right).
	$$
	We can do the following transformation
	$$
	(1+|a|)\varphi_{0,1}(a)=(1+|a|)\frac{1}{\sqrt{2\pi}}\exp\{-a^2/2\}\leq c e^{-a^2/4},
	$$
	where $c=\frac{1}{\sqrt{2\pi}}\sup\limits_{x}(1+x)\exp\{-x^2/4\}$. Then
	\begin{equation}\label{chern_q_v_eq1_part2}
	|q_v(v,z)| \leq (1+\sigma)\E\left( \frac{P(v,a)}{K(v,a)}\cdot\frac{c e^{-a^2/4}}{K(v,a)}\right).
	\end{equation}
	Next we will estimate separately $P(v,a)/K(v,a)$ and $1/K(v,a)$. So,
	$$
		\frac{P(v,a)}{K(v,a)}=\frac{P(v,a)}{\sigma\int\limits_0^v u P(u,a)du}=
		\frac{1}{\sigma\int\limits_0^v u \frac{P(u,a)}{P(v,a)}du}.	
	$$
	Consider
	\begin{align*}
		\frac{P(u,a)}{P(v,a)}&=\exp\{\sigma W_u-\sigma uW_1-\sigma^2 u/2+\sigma ua 
		-\sigma W_v+\sigma vW_1+\sigma^2 v/2-\sigma va\}
		\\[5mm]
		&=\exp\{-\sigma (W_v-W_u)+\sigma(v-u)W_1+\sigma^2(v-u)/2-\sigma(v-u)a\}.
	\end{align*}
	Then
	$$
	\sigma\int\limits_0^v u \frac{P(u,a)}{P(v,a)}du=
	\sigma\int\limits_0^v u \exp\{-\sigma (W_v-W_u)+\sigma(v-u)W_1+\sigma^2(v-u)/2-\sigma(v-u)a\}du	
	$$
	Make the change of variable $t=v-u$ 
	$$
	\sigma\int\limits_0^v u \frac{P(u,a)}{P(v,a)}du=
	\sigma\int\limits_0^v (v-t) \exp\{-\sigma \bar{W_t}+\sigma t W_1+\sigma^2 t/2-\sigma t a\}dt,
	$$
	here $\bar{W_t}=W_v-W_{v-t}$. Next we make the change of variable $s=t\sigma^2$ to use the scale invariance property of Wiener process.
	\begin{align*}
	\sigma\int\limits_0^v u \frac{P(u,a)}{P(v,a)}du=\frac{1}{\sigma^3}
	\int\limits_0^{\sigma^2 v} (\sigma^2v-s)\exp\{- W^*_s+sW_1/\sigma + s/2- s a/\sigma \} ds,
	\end{align*}
	where $W^*_s=\sigma\bar{W}_{s/\sigma^2}$. Find the lower bound for the last expression.
	\begin{align*}
	\sigma\int\limits_0^v u \frac{P(u,a)}{P(v,a)}du &\geq \frac{1}{\sigma^3}
	\int\limits_0^{v_*} (v_*-s)ds \exp\{-\max\limits_{0\leq s\leq 1} W^*_s-|W_1| -|a|\}
	\\[5mm]
	&\geq \frac{v_*^2}{\sigma^3} \exp\{-\max\limits_{0\leq s\leq 1} W^*_s-|W_1| -|a|\}.
	\end{align*}
	Thus
	$$
		\frac{P(v,a)}{K(v,a)}\leq \frac{\sigma^3}{v_*^2}\exp\{-\max\limits_{0\leq s\leq 1} W^*_s-|W_1| -|a|\}
	$$
	
	Consider analogically
	\begin{align*}
	K(v,a)&=\sigma\int\limits_0^v u \exp\{\sigma W_u-\sigma u W_1-\sigma^2 u/2+\sigma u a\}du
	\\[3mm]
	&=\frac{1}{\sigma^3}\int\limits_0^{\sigma^2 v} s\exp\{\hat{W}_s-sW_1/\sigma - s/2+s a/\sigma \} ds
	\\[3mm]
	&\geq \frac{v_*^2}{\sigma^3} \exp\{-\max\limits_{0\leq s\leq 1} \hat{W}_s-|W_1| -|a|\},
	\end{align*}
	where $\hat{W}_s=\sigma W_{s/\sigma^2}$. Finally we obatin
	
	\begin{align*}
	|q_v(v,z)| &\leq (1+\sigma)\E\left( \frac{P(v,a)}{K(v,a)}\cdot\frac{c e^{-a^2/4}}{K(v,a)}\right)
	\\[3mm]
	&\leq \frac{\sigma^7}{v_*^4} \E \exp\{\max\limits_{0\leq s\leq 1}W^*_s+2|W_1|+2|a|+
	\max\limits_{0\leq s\leq 1} \hat{W}_s-a^2/4\}
	\end{align*}
	Introduce notation $\gamma^*=\gamma_1+\gamma_2+\gamma_3$ with conponents
	$$
	\gamma_1=\max\limits_{0\leq s\leq 1}W^*_s\;\;\;\;
	\gamma_2=\max\limits_{0\leq s\leq 1} \hat{W}_s\;\;\;\;
	\gamma_3=2|W_1|
	$$
	It's clear that 
	$$
	\E\exp\{N\gamma^*\}<+\infty,\;\;\;\forall N
	$$ 
	since
	\begin{align*}
	\E e^{\gamma_1}e^{\gamma_2}e^{\gamma_3}&\leq \left(\E e^{2\gamma_1}\right)^{1/2}
	\left(\E e^{2\gamma_2}e^{2\gamma_3}\right)^{1/2}
	\\[3mm]
	&\leq \left(\E e^{2\gamma_1}\right)^{1/2}\left(\E e^{4\gamma_2}\right)^{1/4}
	\left(\E e^{4\gamma_3}\right)^{1/4}.
	\end{align*}
	and  $\E\exp\{N\max |W_t|\}<+\infty$. Then for $c^*=c_1c\left(\E e^{2\gamma_1}\right)^{1/2}\left(\E e^{4\gamma_2}\right)^{1/4}\left(\E e^{4\gamma_3}\right)^{1/4}$
	\begin{align*}
	|q_v(v,z)| &\leq \frac{\sigma^7 c^*}{v_*^4} \E \exp\{\gamma^*+2|a|-a^2/8-a^2/8\}
	\\[3mm]
	&\leq \frac{\sigma^7 \hat{c}}{v_*^4} \E \exp\{\gamma^*-a^2/8\},
	\end{align*}
	here $\hat{c}=c^*\exp\{\sup\limits_{x}(2x-x^2/8)\}$.
	
	Next we obtain the lower bound for $a(v,z)$ similar to the proof of Proposition\ref{shish_density_estim_part1}
	\begin{equation}\label{estimate_of a_part2}
			|a(v,z)|\geq \frac{1}{\sigma v}\left(\ln(z/v)-\sigma\sqrt{v}\beta_*\right),
		\end{equation}
	where $\beta_*=\max\limits_{0\leq u\leq v}|W_u|/\sqrt{v}+|W_1|$ and $\E\exp\{N\beta_*<+\infty\}$.
	Also we represent the expectation as	
	 	\begin{align*}\label{ineq with indicators_part2}
	 		\E\exp\{\gamma^*-a^2(v,z)/8\}&=\E\exp\{\gamma^*-a^2(v,z)/8\}
	 		(\Chi_{\{\beta_*\leq L\}}+\Chi_{\{\beta_*> L\}})\\
	 		&\leq \E\exp\{\gamma^*-a^2(v,z)/8\}\Chi_{\{\beta_*\leq L\}}+
	 		\E\exp\{\gamma^*\}\Chi_{\{\beta_*> L\}}\\
	 		&\leq \E e^{\gamma^*}\exp\{-a^2(v,z)/8\}\Chi_{\{\beta_*\leq L\}}+
	 		(\E e^{2\gamma^*})^{1/2}(\P(\beta_*> L))^{1/2}
	 	\end{align*}
	Let $c_2=\max (\E e^{\gamma^*},(\E e^{2\gamma^*})^{1/2})$. Using Markov's inequality we obtain
		\[
			\P(\beta_*> L)=\P(e^{\delta_*\beta^2_*}>e^{\delta_* L^2})\leq 
			\exp\{-\delta_* L^2\}\E\exp\{\delta_*\beta^2_*\}=c_3^2\exp\{-\delta_* L^2\}.
		\]
	Then
	\begin{align*}
	 		\E\exp\{\gamma^*-a^2(v,z)/8\}\leq c_2(\exp\{-a^2(v,z)/8\}\Chi_{\{\beta_*\leq L\}}+
	 		c_3\exp\{-\delta_* L^2/2\})
	 \end{align*}
	If $\beta_*\leq L$ then inequality 	(\ref{estimate_of a_part2}) will take the form
		\[
			|a(v,z)|\geq \frac{1}{\sigma v}\left(\ln(z/v)-\sigma\sqrt{v}\beta_*\right)
			\geq \frac{1}{\sigma v}\left(\ln(z/v)-\sigma\sqrt{v} L\right).
		\]
	The constant $L$ must be chosen so that $\ln(z/v)-\sigma L$>0. Let
		\[
			L=\frac{1}{2\sigma\sqrt{v}}\ln\left( \frac{z}{v} \right).
		\]
	Then 
	$$
	|a(v,z)|>\frac{1}{2\sigma v}\ln(z/v)
	$$ 
	and 
		\begin{align*}
		 \E\exp\{\gamma^*-a^2(v,z)/8\}&\leq c_2\left(\exp\left\lbrace -\frac{1}{32\sigma^2 v^2}(\ln(z/v))^2\right\rbrace+c_3\exp\left\lbrace -\frac{\delta_*}{8\sigma^2 v} (\ln(z/v))^2\right\rbrace\right)
		 \\
		 &\leq c_2(1+c_3)\exp\left\lbrace -\frac{\delta_*}{8\sigma^2 v} (\ln(z/v))^2\right\rbrace .
		\end{align*}
	Thus for the constants $\tilde{c}=c_2(1+c_3)\hat{c}$ and $\kappa=\delta_*/8$ which not depend on $\sigma$ we have the following estimate for the derivative of density w.r.t. $v$
	$$
		q_v(v,z)\leq \frac{\tilde{c}\sigma^7}{v_*^4}\exp\left\lbrace -\frac{\kappa}{\sigma^2 v} (\ln(z/v))^2\right\rbrace.
	$$
	
	\end{proof}


%
%



\end{document}